	\documentclass[
		paper=a4,11pt,american,pagesize,%
		abstract=true,headinclude=true,headlines=0,footinclude=true,footlines=-5,twoside=semi,%
		DIV=12,%
	]{scrartcl}
	\def\docclass{koma}
	\def\version{arxiv}
	\def\draftmode{false} % hide debug info & notes-to-self/todo etc.
\usepackage[utf8]{inputenc}
\makeatletter
\usepackage{xifthen}

\newcommand\iflipics[2]{\ifthenelse{\equal{\docclass}{lipics}}{#1}{#2}}
\newcommand\ifkoma[2]{\ifthenelse{\equal{\docclass}{koma}}{#1}{#2}}
\newcommand\ifieee[2]{\ifthenelse{\equal{\docclass}{ieee}}{#1}{#2}}
\newcommand\ifsiam[2]{\ifthenelse{\equal{\docclass}{siam}}{#1}{#2}}
\newcommand\ifsiamsingle[2]{\ifthenelse{\equal{\docclass}{siam-single}}{#1}{#2}}
\newcommand\ifmysiam[2]{\ifthenelse{\equal{\docclass}{my-siam}}{#1}{#2}}
\newcommand\ifacm[2]{\ifthenelse{\equal{\docclass}{acm}}{#1}{#2}}
\newcommand\ifdcc[2]{\ifthenelse{\equal{\docclass}{dcc}}{#1}{#2}}
\newcommand\ifspringerjournal[2]{\ifthenelse{\equal{\docclass}{springer-journal}}{#1}{#2}}
\newcommand\iflncs[2]{\ifthenelse{\equal{\docclass}{lncs}}{#1}{#2}}
\ifthenelse{ \equal{\docclass}{lipics} \OR \equal{\docclass}{koma} \OR \equal{\docclass}{ieee} \OR \equal{\docclass}{siam} \OR \equal{\docclass}{my-siam} \OR \equal{\docclass}{acm} \OR \equal{\docclass}{dcc} \OR \equal{\docclass}{springer-journal} \OR \equal{\docclass}{lncs} }{
	\PackageInfo{paper}{Building paper with docclass = \docclass} % all good
}{
	\PackageWarning{paper}{docclass = "\docclass", but must be one of "lipics", "koma", "ieee", "siam", "siam-single", "my-siam", "acm", "dcc", "springer-journal", "lncs"}
}

\newcommand\ifmanuscript[2]{\ifthenelse{\equal{\version}{manuscript}}{#1}{#2}}
\newcommand\ifarxiv[2]{\ifthenelse{\equal{\version}{arxiv}}{#1}{#2}}
\newcommand\ifsubmission[2]{\ifthenelse{\equal{\version}{submission}}{#1}{#2}}
\newcommand\ifproceedings[2]{\ifthenelse{\equal{\version}{proceedings}}{#1}{#2}}
\ifthenelse{ 
	\equal{\version}{manuscript} 
	\OR \equal{\version}{arxiv} 
	\OR \equal{\version}{submission} 
	\OR \equal{\version}{proceedings} 
}{
	\PackageInfo{paper}{Building paper version = \version} % all good
}{
	\PackageWarning{paper}{version = "\version", but must be one of "manuscript", "arxiv", "submission", "proceedings"}
}

\newcommand\ifdraft[2]{\ifthenelse{\equal{\draftmode}{true}}{#1}{#2}}
\ifthenelse{ \equal{\draftmode}{true} \OR \equal{\draftmode}{false} }{
	\PackageInfo{paper}{Building paper with draftmode = \draftmode} % all good
}{
	\PackageWarning{paper}{draftmode = "\draftmode", but must be "true" or "false"}
}

\usepackage[T1]{fontenc}
\ifsiam{
	\usepackage{lmodern}
	\usepackage{slantsc}
}{}
\ifsiamsingle{
	\usepackage{lmodern}
	\usepackage{slantsc}
}{}
\ifmysiam{
	\usepackage{lmodern}
	\usepackage{slantsc}
}{}
\ifkoma{
	\usepackage{lmodern}
	\usepackage{slantsc} % Use EC for bold SC
}{}
\iflipics{
	\usepackage{lmodern}
	\usepackage{slantsc}
}{}
\ifdcc{
	\usepackage{lmodern}
	\usepackage{slantsc}
}{}
\ifspringerjournal{
	\usepackage{lmodern}
	\usepackage{slantsc}
	\usepackage{xcolor}
}{}
\iflncs{
	\usepackage{lmodern}
	\usepackage{slantsc}
	\usepackage{xcolor}
}{}

\usepackage{babel}
\input{ushyphex.tex} % American English hyphenation exceptions

\usepackage{array,multicol,multirow}
\ifieee{
	\usepackage[cmex10]{amsmath,mathtools}
	\usepackage{amsfonts,amssymb}
}{}
\ifkoma{
	\usepackage{amsmath,amsfonts,amssymb,mathtools}
}{}
\iflipics{
	\usepackage{amsmath,amsfonts,amssymb,mathtools}
}{}
\ifsiam{
	\usepackage{amsmath,amsfonts,amssymb,mathtools}
}{}
\ifsiamsingle{
	\usepackage{amsmath,amsfonts,amssymb,mathtools}
}{}
\ifmysiam{
	\usepackage{amsmath,amsfonts,amssymb,mathtools}
}{}
\ifacm{
	\usepackage{mathtools}
}{}
\ifdcc{
	\usepackage{amsmath,amsfonts,amssymb,mathtools}
}{}
\ifspringerjournal{
	\usepackage{amsmath,amsfonts,amssymb,mathtools}
}{}
\iflncs{
	\usepackage{amsmath,amsfonts,amssymb,mathtools}
}{}

\usepackage{mleftright}\mleftright % redefine left and right to not add space around braces
\usepackage{relsize,xspace,booktabs,adjustbox,needspace,pbox,relsize,makecell}
\ifieee{
		
		\usepackage{enumitem}
}{
	\ifspringerjournal{
		\usepackage{enumitem}
		\setlist{topsep=\medskipamount}
	}{
		\usepackage{enumitem}
	}
}
\usepackage{graphicx}

\ifacm{}{
	\usepackage{colonequals}
}

\usepackage{wref}

\ifsiam{
	\usepackage{ltexpprt}
}{}
\ifsiamsingle{
	\usepackage{ltexpprt}
}{}

\newdimen\makeboxdimen
\newcommand\makeboxlike[3][l]{%
\setbox0=\hbox{#2}%
\global\makeboxdimen=\wd0%
\setbox1=\hbox{\makebox[\makeboxdimen][#1]{%
\makebox[0pt][#1]{#3}%
}}%
\ht1=\ht0%
\dp1=\dp0%
\box1%	
}

\newcommand\like[3][c]{%
	\mathchoice{%dis
		\makeboxlike[#1]{%
			\ensuremath{\displaystyle\relax#2}%
		}{%
			\ensuremath{\displaystyle\relax#3}%
		}%
	}{%txt
		\makeboxlike[#1]{%
			\ensuremath{\textstyle\relax#2}%
		}{%
			\ensuremath{\textstyle\relax#3}%
		}%
	}{%sc
		\makeboxlike[#1]{%
			\ensuremath{\scriptstyle\relax#2}%
		}{%
			\ensuremath{\scriptstyle\relax#3}%
		}%
	}{%scsc
		\makeboxlike[#1]{%
			\ensuremath{\scriptscriptstyle\relax#2}%
		}{%
			\ensuremath{\scriptscriptstyle\relax#3}%
		}%
	}
}

\newcommand\plaincenter[1]{%
	\mbox{}\hfill#1\hfill\mbox{}%
}

\ifthenelse{\equal{\docclass}{koma} \OR \equal{\docclass}{my-siam}}{
	\setlength\parindent{1.5em}
	\usepackage[headsepline]{scrlayer-scrpage}
	\pagestyle{scrheadings}
	\clearscrheadfoot
	\AtBeginDocument{%
		\automark[section]{}%
	}
	\ohead{\pagemark}
	\rehead{\mytitle}
	\lohead{\headmark}
	\addtokomafont{caption}{\sffamily\small}
	\addtokomafont{captionlabel}{\sffamily\textbf}
	\setcapmargin{2em}
}{}
\ifmysiam{
	\setcapmargin{1em}
	\setcapindent{0em}
}{}
\ifmysiam{
	\setlength\parskip{0pt}
	\RedeclareSectionCommand[
		beforeskip=-1.25\baselineskip,
		afterskip=0.75\baselineskip,
	]{section}
	\RedeclareSectionCommand[
		beforeskip=-1\baselineskip,
		afterskip=-1.5em,
	]{subsection}
	\RedeclareSectionCommand[
		beforeskip=-1\baselineskip,
		afterskip=-1.5em,
	]{subsubsection}
	\RedeclareSectionCommand[
		beforeskip=-.25\baselineskip,
		indent=1.5em,
		afterskip=-1em,
	]{paragraph}
}{}
\ifdcc{
	\ifproceedings{}{
		\pagestyle{plain}
		\setlength{\footskip}{5ex}
	}
}{}

\AtBeginDocument{%
	\let\mytitle\@title%
}

\ifsiam{
	\usepackage[bibtex-alpha]{url-doi-arxiv}
}{}
\ifsiamsingle{
	\usepackage[bibtex-alpha]{url-doi-arxiv}
}{}
\ifmysiam{
	\usepackage[bibtex-alpha]{url-doi-arxiv}
}{}
\ifkoma{
	\usepackage[bibtex-alpha]{url-doi-arxiv}
}{}
\iflipics{
	\usepackage[bibtex]{url-doi-arxiv}
}{}
\ifdcc{
	\usepackage[bibtex-alpha]{url-doi-arxiv}
}{}
\ifspringerjournal{
	\usepackage[bibtex-alpha]{url-doi-arxiv}
}{}
\iflncs{
	\usepackage[bibtex-alpha]{url-doi-arxiv}
}{}

\let\oldthebibliography\thebibliography
\renewcommand\thebibliography[1]{%
	\oldthebibliography{#1}%
	\pdfbookmark[1]{References}{}%
}

\usepackage{lscape} % sideways figures
\usepackage{rotating}

\ifkoma{
	\usepackage{float}
	\floatstyle{plain}
	\usepackage{newfloat}
	\DeclareFloatingEnvironment[%
			name=Algorithm,%
			placement=thb,%
		]{algorithm}
}{}
\iflipics{
	\usepackage{newfloat}
	\DeclareFloatingEnvironment[%
			name=Algorithm,%
			placement=thb,%
		]{algorithm}
}{}
\ifmysiam{
	\usepackage{newfloat}
	\DeclareFloatingEnvironment[%
			name=Algorithm,%
			placement=thb,%
		]{algorithm}
}

\ifdcc{
	\usepackage{float}
	\usepackage[font={small},labelfont=bf]{caption}
}{}
\ifmysiam{
	\setcounter{topnumber}{3}
	\setcounter{bottomnumber}{3}
	\setcounter{totalnumber}{3}     % 2 may work better
	\setcounter{dbltopnumber}{3}    % for 2-column pages
}{}
\ifsiam{
	\setcounter{topnumber}{3}
	\setcounter{bottomnumber}{3}
	\setcounter{totalnumber}{3}     % 2 may work better
	\setcounter{dbltopnumber}{3}    % for 2-column pages
}{}
\ifsiamsingle{
	\setcounter{topnumber}{3}
	\setcounter{bottomnumber}{3}
	\setcounter{totalnumber}{3}     % 2 may work better
}{}

\usepackage{dcolumn}

\usepackage{wclrscode}

\usepackage{textcomp} % needed for upquote
\usepackage{listings}

\usepackage{tikz}

\usetikzlibrary{positioning,arrows.meta,fit}
\usetikzlibrary{backgrounds,calc,trees,graphs}
\usetikzlibrary{shapes.geometric,shapes.misc}
\usetikzlibrary{decorations,decorations.pathreplacing}

\pgfdeclarelayer{background}
\pgfsetlayers{background,main}

\usetikzlibrary{external}
\tikzsetexternalprefix{pics/externalized/}
\tikzset{
	external/system call={%
		lualatex \tikzexternalcheckshellescape -halt-on-error %
			-interaction=batchmode -jobname "\image" "\texsource"%
	},
}
\tikzset{external/export=false} % Only export manually marked pictures, enable for costly pictures
\iflipics{
	\newtheorem{fact}[theorem]{Fact}

	\newenvironment{proofof}[1]{%
		\begin{proof}[{{Proof of #1{}}}]%
	}{%
		\end{proof}%
	}
}{}
\ifacm{
	\AtEndPreamble{
		\theoremstyle{acmdefinition}
		\newtheorem{remark}[theorem]{Remark}
		\newtheorem{fact}[theorem]{Fact}
	}
	
	\newenvironment{proofof}[1]{%
		\begin{proof}[{{Proof of #1{}}}]%
	}{%
		\end{proof}%
	}
}{}
\ifsiam{
	\newtheorem{remark}{Remark}
	\newenvironment{proofof}[1]{%
		\begin{proof}[{{#1{}}}]%
	}{%
		\end{proof}%
	}
}{}
\ifsiamsingle{
	\newtheorem{remark}{Remark}
	\newenvironment{proofof}[1]{%
		\begin{proof}[{{#1{}}}]%
	}{%
		\end{proof}%
	}
}{}
\iflncs{
	\spnewtheorem{fact}[theorem]{Fact}{\itshape}{}
	\let\orig@endproof\endproof
	\def\endproof{\qed\orig@endproof}
	\newenvironment{proofof}[1]{%
		\begin{proof}[{{#1{}}}]%
	}{%
		\end{proof}%
	}
}{}
\ifthenelse{%
		\equal{\docclass}{lipics} \OR \equal{\docclass}{siam} \OR 
		\equal{\docclass}{siam-single} \OR \equal{\docclass}{acm} \OR
		\equal{\docclass}{lncs}%
}{}{
	\usepackage[amsmath,hyperref,thmmarks]{ntheorem}
	
	\theoremheaderfont{\sffamily\upshape\bfseries}
	\theorembodyfont{\slshape}
	\theoremseparator{:}
	\newtheoremstyle{proofstyle}%
	  {\item[\theorem@headerfont\hskip\labelsep ##1\theorem@separator]}%
	  {\item[\theorem@headerfont\hskip\labelsep ##3\theorem@separator]}
	
	\theorempreskip{\topsep} % new version of ntheorem

	\theoremsymbol{\adjustbox{scale=.8}{$\triangleleft\mkern-1mu$}}
	
	\newtheorem{theorem}{Theorem}[section]
	
	\theoremstyle{plain}
	\theorempreskip{\topsep}

	\newtheorem{lemma}[theorem]{Lemma}
	
	\newtheorem{corollary}[theorem]{Corollary}

	\theoremstyle{plain}
	\theorembodyfont{\upshape}

	\newtheorem{fact}[theorem]{Fact}

	\theoremsymbol{\raisebox{-.25ex}{$\Box$}}
	\qedsymbol{\raisebox{-.25ex}{$\Box$}}
	
	\theoremstyle{proofstyle}
	\theoremheaderfont{\sffamily\slshape}
	\newtheorem{proof}{Proof}

}

\iflipics{
		\newenvironment{thmenumerate}[2][]{%
			\begin{enumerate}[
				label={\textsf{\textbf{\color{darkgray}{\makebox[\widthof{(a)}][c]{\textup{(\alph*)}}}}}},
				ref={\ref{#2}\kern.1em--\kern.1em(\alph*)},
				itemsep=0pt,
				topsep=.5ex,
				leftmargin=1.75em,
				#1
			]%
		}{%
			\end{enumerate}%
		}
}{
	\ifspringerjournal{
		\newenvironment{thmenumerate}[2][]{%
			\begin{enumerate}[
				label={\makebox[\widthof{(a)}][c]{\textup{(\alph*)}}},
				ref={\ref{#2}\kern.1em--\kern.1em(\alph*)},
				itemsep=0pt,
				topsep=\smallskipamount,
				leftmargin=1.75em,
				#1
			]%
		}{%
			\end{enumerate}%
		}
	}{
		
	}
}

\newcommand*\aka{\mbox{a.\hspace{.2ex}k.\hspace{.2ex}a.}\xspace}

\newcommand\Oh{O}

\usepackage{fixmath}

\newcommand{\ESymbol}{\mathbb{E}}

\newcommand{\ProbSymbol}{\ensuremath{\mathbb{P}}}

\DeclarePairedDelimiterXPP\Prob[1]{\ProbSymbol}[]{}{%
	#1%
}
\DeclarePairedDelimiterXPP\E[1]{\ESymbol}[]{}{%
	#1%
}
\DeclarePairedDelimiterXPP\Eover[2]{\ESymbol_{#1}}[]{}{%
	#2%
}
\DeclarePairedDelimiterXPP\ProbIn[2]{\ProbSymbol_{#1}}[]{}{%
	#2%
}
\providecommand{\Prob}{} % hack for syntax highlighting ...
\providecommand{\ProbIn}{} % hack for syntax highlighting ...
\providecommand{\E}{} % hack for syntax highlighting ...
\providecommand{\Eover}{} % hack for syntax highlighting ...

\newcommand{\surroundedmath}[3]{% #1=mathrel/mathbin/etc #2=spacing #3=symbol
	\mathchoice{%display
		#1{#2{#3}#2}%
	}{%text
		#1{#3}%
	}{%script
		#1{#3}%
	}{%scriptsript
		#1{#3}%
	}%
}
\newcommand\rel[1]{\surroundedmath{\mathrel}{\:}{#1}}
\newcommand\wrel[1]{\surroundedmath{\mathrel}{\;}{#1}}
\newcommand\wwrel[1]{\surroundedmath{\mathrel}{\;\;}{#1}}

\iflipics{}{
	\makeatletter
	\let\oldalign\align
	\let\endoldalign\endalign
	\renewenvironment{align}{%
		\begingroup%
		\let\oldhalign\halign
		\def\halign{%
			\let\oldbreak\\%
			\def\nonnumberbreak{\nonumber\oldbreak*}%
			\def\\{%
				\@ifstar{\nonnumberbreak}{\oldbreak}%
			}%
			\oldhalign%
		}
		\oldalign%
	}{%
		\endoldalign%
		\endgroup%
	}
}
\newcommand*\numberthis[1][]{\stepcounter{equation}\tag{\theequation}}

\allowdisplaybreaks[3]

\newcommand\splitaftercomma[1]{%
  \begingroup
  \begingroup\lccode`~=`, \lowercase{\endgroup
    \edef~{\mathchar\the\mathcode`, \penalty0 \noexpand\hspace{0pt plus .25em}}%
  }\mathcode`,="8000 #1%
  \endgroup
}

\def\mydots{\xleaders\hbox to.5em{\hfill.\hfill}\hfill}
\newlength\tmpLenNotations

\ifdraft{%
	\iflipics{
		\usepackage{lineno} 
	}{
		\usepackage[switch]{lineno} 
	}
	\linenumbers
	\overfullrule=6mm
	
	\usepackage[color,notref,notcite]{showkeys}
	\definecolor{refkey}{gray}{.99}
	\colorlet{labelkey}{green!60!black!60}
	
	\usepackage[inline,nolabel]{showlabels}

	\showlabels{cite}
	\showlabels{citealt}
	\showlabels{citealp}
	\showlabels{citet}
	\showlabels{Citet}
	\showlabels{citep}
	\showlabels{citeauthor}
	\showlabels{Citeauthor}
	\showlabels{citefullauthor}
	\showlabels{citeyear}
	\showlabels{citeyearpar}
	\showlabels{wref}
	\showlabels{wpref}
	\showlabels{wtpref}
	\showlabels{wildref}
	\showlabels{wildpageref}
	\showlabels{wildtpageref}
}{}

\iflipics{
	\ifmanuscript{\hideLIPIcs}{}
	\ifarxiv{\hideLIPIcs}{}
	\ifsubmission{}{\nolinenumbers}
}{}

\iflncs{
    \ifsubmission{
        \usepackage[switch]{lineno}
        \linenumbers
    }{}
}{}

\ifdraft{}{%
	\usepackage{microtype}
}

\hypersetup{
	final,
	unicode=true, 
	bookmarks=true,
	bookmarksnumbered=true,
	bookmarksdepth=2,
	bookmarksopen=true,
	breaklinks=true,
	hidelinks,
}

\newsavebox\tmpbox

\ifdcc{
	\renewcommand\paragraph{\@startsection{paragraph}{4}{\parindent}%\z@}%
	                                      {\smallskipamount}%{3.25ex \@plus1ex \@minus.2ex}%
	                                      {-1em}%
	                                      {\normalfont\normalsize\bfseries}}
}{}
\iflipics{
	\let\oldparagraph\paragraph
	\renewcommand\paragraph[1]{%
		\oldparagraph*{#1}
	}
}{
	\let\oldparagraph\paragraph
	\renewcommand\paragraph[1]{%
		\oldparagraph{#1.}
	}
}

\ifmysiam{
	\let\oldsubsection\subsection
	\renewcommand\subsection[1]{%
		\oldsubsection{#1.}%
	}
	\let\oldsubsubsection\subsubsection
	\renewcommand\subsubsection[1]{%
		\oldsubsubsection{#1.}%
	}
}{}
\ifsiam{
	\let\oldsubsection\subsection
	\renewcommand\subsection[1]{%
		\oldsubsection{#1.}%
	}
	\let\oldsubsubsection\subsubsection
	\renewcommand\subsubsection[1]{%
		\oldsubsubsection{#1.}%
	}
}{}
\ifsiamsingle{
	\let\oldsubsection\subsection
	\renewcommand\subsection[1]{%
		\oldsubsection{#1.}%
	}
	\let\oldsubsubsection\subsubsection
	\renewcommand\subsubsection[1]{%
		\oldsubsubsection{#1.}%
	}
}{}

\let\epsilon\varepsilon

\def\myacknowledgements{}
\ifkoma{
	\newcommand\acknowledgements[1]{\def\myacknowledgements{\paragraph{Acknowledgements}#1}}
}{}
\ifieee{
	\newcommand\acknowledgements[1]{\def\myacknowledgements{\section*{Acknowledgement}#1}}
}{}
\ifsiam{
	\newcommand\acknowledgements[1]{\def\myacknowledgements{\section*{Acknowledgement}#1}}
}{}
\ifsiamsingle{
	\newcommand\acknowledgements[1]{\def\myacknowledgements{\section*{Acknowledgement}#1}}
}{}
\ifmysiam{
	\newcommand\acknowledgements[1]{\def\myacknowledgements{\section*{Acknowledgement}#1}}
}{}
\ifdcc{
	\newcommand\acknowledgements[1]{\def\myacknowledgements{\section*{Acknowledgement}#1}}
}{}
\ifacm{
	\newcommand\acknowledgements[1]{\def\myacknowledgements{
		\section*{Acknowledgement}#1%
	}}
}{}
\ifspringerjournal{
	\newcommand\acknowledgements[1]{\def\myacknowledgements{\bmhead{Acknowledgements}#1}}
}{}
\iflncs{
	\newcommand\acknowledgements[1]{\def\myacknowledgements{%
		\begin{credits}
		\subsubsection{\ackname} #1
		\end{credits}
	}}
}{}

\setlist[description]{font=\boldmath}

\makeatother

\usepackage{hyperref}
\usepackage{orcidlink}
\ifacm{
		\title[Short Title]{My Long Paper Title}
}{}
\ifspringerjournal{
	\title[Short Title]{My Long Paper Title}
}{
	\title{Partition-based Simple Heaps}
}

\ifthenelse{\NOT \equal{\docclass}{lipics} \AND \NOT \equal{\docclass}{acm} \AND \NOT \equal{\docclass}{springer-journal} \AND \NOT \equal{\docclass}{lncs} }{
	\newcommand\email[1]{\texttt{#1}}
	\author{%
		Gerth Stølting Brodal%
		\footnote{%
			Aarhus University, Denmark,
		    \email{\{gerth,rysgaard\}@cs.au.dk}
		}
		\,\orcidlink{0000-0001-9054-915X}
		\and
		John Iacono%
		\footnote{%
			Université libre de Bruxelles, Ixelles, Belgium, 
			\email{john.iacono@ulb.be}
		}
		\,\orcidlink{0000-0001-8885-8172}
		\and
		Casper Moldrup Rysgaard\footnotemark[1]
		\,\orcidlink{0000-0002-3989-123X}%
		\and 
		Sebastian Wild%
		\footnote{%
			University of Marburg, Germany, 
			\email{wild\,@\,informatik.uni-marburg.de};
			University of Liverpool, UK
		}
		\,\orcidlink{0000-0002-6061-9177}%
	}
	
	\date{\small\today}
}{}

\acknowledgements{%
	We thank our reviewers for their helpful comments. 
	The second author's work is supported by the Fonds de la Recherche Scientifique --- FNRS.
	The last author is partly supported by EPSRC grant EP/X039447/1.
}

\DeclarePairedDelimiter{\lrParens}{(}{)} %mathtools
\newcommand{\lrSize}[1]{\left| #1 \right|}

\renewcommand{\Oh}[2][*]{%
    \mathcal{O}\ifthenelse{\equal{#1}{*}}{\lrParens*{#2}}{\lrParens[#1]{#2}}%
}
\newcommand{\OhTheta}[2][*]{%
    \Theta\ifthenelse{\equal{#1}{*}}{\lrParens*{#2}}{\lrParens[#1]{#2}}%
}
\newcommand{\OhMega}[2][*]{%
    \Omega\ifthenelse{\equal{#1}{*}}{\lrParens*{#2}}{\lrParens[#1]{#2}}%
}

\let\phi\varphi
\newcommand{\after}[1]{\like{#1}{\overline{#1}}}

\newcommand{\Int}[1]{S_{#1}}
\newcommand{\IntAfter}[1]{\after{S}_{#1}}
\newcommand{\outBase}[2][*]{%
    \mathbf{s}%
    \ifthenelse{\equal{#1}{*}}{\lrParens*{#2}}{\lrParens[#1]{#2}}%
}
\newcommand{\out}[2][*]{%
    \outBase[#1]{\Int{#2}}%
}

\newcommand{\INSERT}{\textrm{\textsc{Insert}}\xspace}
\newcommand{\DECREASEKEY}{\textrm{\textsc{Decrease-Key}}\xspace}
\newcommand{\DELETEMIN}{\textrm{\textsc{Delete-Min}}\xspace}
\newcommand{\INSERTplain}{Insert\xspace}
\newcommand{\DECREASEKEYplain}{Decrease-Key\xspace}
\newcommand{\DELETEMINplain}{Delete-Min\xspace}
\newcommand{\maxzero}[1]{\max\left\{0,\,#1\right\}}
\newcommand{\newkey}{\id{key}}
\newcommand{\Sinv}{S}

\newcommand{\ARXIV}[1]{\ifproceedings{}{#1}}
\newcommand{\LATIN}[1]{\ifproceedings{#1}{}}

\newcommand\CSHeaps{LP\xspace}
\newcommand\CSHeapsLong{lazy partition heap\xspace}

\newcommand\Pot[1]{\phi_{#1}}
\newcommand\PotAfter[1]{\after{\phi}_{#1}}

\begin{document}

\ifacm{}{\maketitle} %

\begin{abstract}
    We introduce a new family of priority-queue data structures: \emph{partition-based simple heaps}. 
    The structures consist of $\Oh{\lg n}$ doubly-linked lists; order is enforced among data in different lists, but the individual lists are unordered. Our structures have amortized $\Oh{\lg n}$ time extract-min and amortized $\Oh{\lg \lg n}$ time insert and decrease-key. 
    
    The structures require nothing beyond binary search over $\Oh{\lg n}$ elements, as well as binary partitions and concatenations of linked lists in natural ways as the linked lists get too big or small. We present three different ways that these lists can be maintained in order to obtain
    the stated amortized running times.
\end{abstract}

\ifacm{%
	\maketitle%
}{}

\ifarxiv{\bigskip}{}

\section{Introduction}
\label{sec:introduction}

The \emph{priority queue} is a fundamental comparison-based abstract data type used in numerous efficient algorithms, such as discrete event simulation, Dijkstra's shortest path algorithm and Prim's minimum spanning tree algorithm.
A priority queue (\aka \emph{heap}) stores a set of elements and has to implement the following operations, where we assume each element has an associated key from a comparison-based universe:

\begin{itemize}

\item $\INSERT(e)$: Add an element $e$ to the the heap and return a pointer $\id{ptr}$ to it.

\item $\DELETEMIN()$: Remove and return the element in heap with smallest key.

\item $\DECREASEKEY(\id{ptr}, \newkey)$: Change the key of the element pointed to by \id{ptr} in the heap  to a smaller key $\newkey$.

\end{itemize}

A standard undergraduate computer-science curriculum discusses binary heaps~\cite{Williams64}, 
which realize all operations in $\Oh{\lg n}$ time%
\footnote{Here and throughout we let $\lg$ denote the binary logarithm.}
when the data structure currently contains $n$ elements (as would a balanced binary search tree).  However, much faster updates for priority queues are possible:
The most efficient heaps such as (strict) Fibonacci heaps~\cite{FredmanTarjan1987,DBLP:journals/talg/BrodalLT25} achieve \emph{constant} (amortized) time for \INSERT and \DECREASEKEY and $\Oh{\lg n}$ time for \DELETEMIN.

While optimal in theory, these \emph{efficient priority queues} have some downsides.
First, their structure and analysis are less intuitive to teach and implement.
Second, the constant factors in their running time and space usage are substantially higher than for binary heaps;
not least because they inherently rely on pointer-based representations of trees, with the ensuing poor locality of reference.

The desire to simplify efficient priority queues and to make them more efficient in practice has resulted in a number of alternative data structures, which we review briefly in Section~\ref{sec:related-work}. 
They all retain somewhat complicated structures and invariants, are nontrivial to analyze, or require $\OhMega{\lg n}$ time for updates.

In this paper, we present simple priority-queue implementations solely based on linear data structures,
showing that simple quicksort-style partitioning of elements into unsorted sets suffices for amortized $\Oh{\lg\lg n}$ time {\INSERT} and {\DECREASEKEY} operations and amortized $\Oh{\lg n}$ time {\DELETEMIN} operations~-- provided it is governed by a carefully chosen \emph{``pivot-forgetting rule''} to keep the number of sets in $\Oh{\lg n}$.
We present several such rules, demonstrating a rich design space. 
This brings us much closer to the efficient priority queues above, but with entirely elementary techniques and simple designs.
As a case in point, for the arguably simplest rule (``\CSHeapsLong{}s''), we give a fully self-contained description including its amortized analysis suitable for teaching in undergraduate algorithms classes in Sections~\ref{sec:partition-based-heaps} and~\ref{sec:llm-heaps}.

\subsection{Previous Work}
\label{sec:related-work}

See \cite{DBLP:conf/birthday/Brodal13} for a full survey of priority queues. 
We begin the history of priority queues with fast \DECREASEKEY\ which begins with the Fibonacci heap \cite{FredmanTarjan1987}. They support constant-amortized time \DECREASEKEY\ and \INSERT, and $\Oh{\lg n}$ amortized-time \DELETEMIN. 
 Since then a variety of other heaps with fast \DECREASEKEY\ have been introduced:

\begin{itemize}

\item The pairing heap \cite{DBLP:journals/algorithmica/FredmanSST86}, which is a self-adjusting structure modeled after splay trees, and which we do not know the asymptotic amortized running time of \DECREASEKEY. 
Only $\Oh{2^{\sqrt{\lg \lg n}}}$ \cite{DBLP:conf/focs/Pettie05} and $\Omega(\lg \lg n)$ \cite{DBLP:journals/jacm/Fredman99} are known. A multipass variant has been shown to have an $\Oh{\lg \lg n \cdot \lg \lg \lg n}$ amortized-time \DECREASEKEY~\cite{DBLP:conf/soda/SinnamonT23}.

\item The strict Fibonacci heap, which has worst-case running times~\cite{DBLP:journals/talg/BrodalLT25}.

\item Chan's quake heaps, invented to be a simpler alternative to Fibonacci heap with $\Oh{1}$ amortized time \DECREASEKEY~\cite{DBLP:conf/birthday/Chan13}.

\item Rank-pairing heaps \cite{DBLP:conf/esa/HaeuplerST09} which are another alternative with $\Oh{1}$ amortized time \DECREASEKEY, but where the structure is closer to the pairing heap.

\item Violation heaps~\cite{DBLP:journals/dmaa/Elmasry10}, 
hollow heaps~\cite{DBLP:journals/talg/HansenKTZ17}, and Fibonacci heaps revisited~\cite{DBLP:journals/corr/KaplanTZ14} are other variants with with $\Oh{1}$ amortized-time \DECREASEKEY.

\item Slim and smooth heaps \cite{DBLP:conf/soda/SinnamonT23a} have $\Oh{\lg \lg n}$ \DECREASEKEY, but are in the pointer-based model with constant indegree.

\end{itemize}

\subsubsection{Lower bounds}

There are two lower bounds for \DECREASEKEY in heaps, one due to Fredman \cite{DBLP:journals/jacm/Fredman99} and the other due to Iacono and \"{O}zkan \cite{DBLP:conf/icalp/IaconoO14}.
Both of these bounds only apply to heap-based priority queues that follow two different particular models, each of which includes pairing heaps, and neither of which applies to our new structures. The Fredman lower bound focuses on the tradeoff between augmented data and the decrease-key time. The Iacono and \"{O}zkan bound focuses on structures that are pointer-based.
Both of these results point to the fact that $\Oh{\lg \lg n}$ is arguably a natural time for heaps, but both only apply to (different) classes of generalized heaps which are quite specific for their adversary arguments to work.

All of the above points to $\Oh{\lg \lg n}$ as the right running time for \DECREASEKEY in a pointer model structure with constant indegree. Slim and smooth heaps achieve this, and pairing heaps are conjectured to. All of the other heaps require either RAM-model tricks or non-constant indegree. For example, Fibonacci heaps need random access into an array of size $\Oh{\lg n}$ in order to combine heaps with the same rank (which is an integer with logarithmic range). One could of course replace this array with a balanced binary search tree, but this would introduce an $\Oh{\log \log n}$ overhead. In~\cite{FredmanTarjan1987}, it is shown how this array can be removed, albeit by using nodes of linear indegree.
 
 \subsubsection{Cache-oblivious inspiration}

We point to the works of Chowdhury and Ramachandran  \cite{DBLP:journals/talg/ChowdhuryR18} and Brodal, Fagerberg, Meyer, and Zeh~\cite{DBLP:conf/swat/BrodalFMZ04} as sources of inspiration as they are devoid of treeology. Both of these, which were designed for the cache-oblivious model, use two sequences of arrays of increasing size; superficially it looks like two copies of our data structure where items are inserted into one set, percolate down, at some point move to the other set, and then percolate up. Being designed for the cache-oblivious model, the lists are stored in arrays. 

\subsubsection{Lazy search trees}
Also related are lazy search trees~\cite{SandlundWild20,SandlundZhang22,RysgaardWild2025}, a data structure that smoothly interpolates 
between priority queues and binary search trees; depending on what queries are used, \INSERT\ can be as fast as in the best priority queues, but arbitrary sorted dictionary queries are supported, as well.
Our structures can be seen as a one-sided instance of the lower two tiers of the original lazy search trees with slightly modified operations.

\subsubsection{Quickheaps}

Other related work takes a more application-focused view of making priority queue implementations fast in practice. While asymptotic performance matters, constant-factor overhead and locality of reference can, for all reasonable input sizes, dwarf the distinction between logarithmic and sublogarithmic complexities.
Quickheaps~\cite{ParedesNavarro2006,NavarroParedes2010} and their ``stronger'' refinements~\cite{NavarroParedesPobleteSanders2011} keep data in a single array and conceptually apply successive quickselect invocations for queries, where past partitioning steps are remembered and skipped.

This works well if updates come in random order; for robustness against adversarial inputs,
stronger quickheaps take inspiration from schemes for balancing binary search trees to ensure good amortized performance in quickheaps: either using randomization or using a form of weight-balanced binary search tree.
Being designed for the much harder problem of keeping an entire binary search tree in shape, the resulting concatenation rules for quickheaps are complicated.
We point out that not only the concatenation rule itself, but also the amortized analysis in~\cite{NavarroParedesPobleteSanders2011} is substantially different from those in this paper.
They also cite favorable performance in the external-memory model as a benefit of quickheaps (in any variant).
Our heaps can also have these same favorable properties.
Note that the authors of quickheaps excluded pointer-based operations, in particular \DECREASEKEY, from their consideration of quickheaps in external memory.

\section{Partition-based Heap Framework}
\label{sec:partition-based-heaps}

Here we describe the generic framework of partition-based heaps common to all three of our data structures in Sections~\ref{sec:llm-heaps}--\ref{sec:exponential}.
We assume that the elements have keys from a totally ordered universe and that elements are ordered with respect to their keys.
If ever duplicate keys are inserted into a partition-based heap, 
standard techniques are used so that all ties are broken consistently.

Partition-based heaps maintain the elements currently stored in the heap
as a partition of $sets$ $S_1, S_2, \ldots S_{\ell}$, for some $\ell$ that is $\Oh{\lg n}$. 
The number of elements currently stored in the heap is $n=|S_1|+\cdots+|S_\ell|$.
While the sets themselves are unordered, if $i<j$, all elements in $S_i$ are less than all elements of $S_j$

Each set $S_i$ is stored as a \emph{doubly-linked list} as constant-time concatenation of sets is vital to all of our structures. For each set, the size $|S_i|$ is explicitly maintained.
As we store all elements in linked lists, we support \emph{stable pointers} to any element inserted.

In addition to the sets and their sizes, 
a partition-based heap stores \emph{pivots} $p_2 \le \cdots \le p_\ell$ which are elements such that $S_i \subseteq [p_i,p_{i+1})$ for $i=1,\ldots,\ell$;
for notational convenience, 
assume $p_1$ and $p_{\ell+1}$ are pivots with keys $-\infty$ and $+\infty$, respectively.
Pivots may be either elements currently or previously in the heap.
See Figure~\ref{fig:partition-heap} for an illustration.

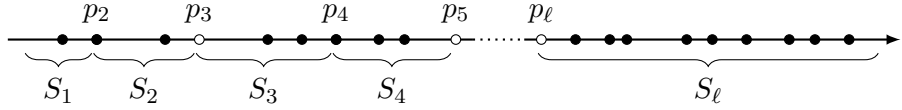
\begin{figure}[t]
    \centering
    \begin{tikzpicture}[scale=0.45]
        \draw[thick, solid] (-0.6,0)--(13,0);
        \draw[thick, dotted] (13,0)--(14.5,0);
        \draw[thick, solid,-latex] (14.5,0)--(25.5,0);
        \foreach \i/\p in {2/2, 3/5, 4/9, 5/12.5, \ell/15} {
            \filldraw[color=black,fill=white,label={above=x}] (\p,0) circle (4pt);
            \node[above] at (\p, 0.2) {$p_\i$};
        }
        \foreach \i/\p in {1,2,4,7,8,9,10.25,11,16,17,17.5,19.25,20,21,22.25,23,24} {
            \filldraw[fill=black] (\p,0) circle (4pt);
        }
        \foreach \i/\l/\r/\s in {1/0/2/1, 2/2/5/2, 3/5/9/3, 4/9/12.5/4, 5/15/25/\ell} {
            \draw [decorate,decoration={brace,amplitude=5pt}] (\r-0.15,-0.4) -- (\l-0.09,-0.4) node[midway,below,below=2mm] {$S_\s$};
        }
    \end{tikzpicture}
     \caption{Illustration of our partition-based heap framework. 
       The pivots $p_2\ldots,p_\ell$ partition the $n$ elements in the heap into sets $\ell$ sets $S_1, S_2, \ldots, S_\ell$.
       Black dots are the $n$ elements in the heap. Pivots can be elements in the heap (black) or elements previously removed from the heap (white). 
       }
     \label{fig:partition-heap}
\end{figure}

To implement each priority queue operation all of our partition-based heaps have the following steps in common:

\begin{itemize}

\item $\INSERT(e)$: 
    Search for $e$ among the $\Oh{\lg n}$ pivots in $\Oh{\lg \lg n}$ time to find the set $S_i$ with $p_i \le e < p_{i+1}$. Insert $e$ into $S_i$ and increment the size $|S_i|$. Return a pointer $\id{ptr}$ to the new node  containing $e$ in the linked list for $S_i$.

\item $\DELETEMIN()$:
    Remove the smallest element from the first nonempty set~$S_i$, and decrement the size~$|S_i|$.

\item $\DECREASEKEY(ptr,\newkey)$: 
    Let $e$ denote the element in the linked-list node $ptr$ points to. Remove this node from its linked list. Binary search on the pivots in  $\Oh{\lg \lg n}$ to find the set $S_i$ that the node was removed from and decrement the size $|S_i|$. Decrease the key of $e$ to $\newkey$, and re-insert the node as in \INSERT, using a second search among the pivots to discover which set to insert it into, insert it at the end of the linked list for the set, and increment the set size.

\end{itemize}

These descriptions are not enough to maintain the crucial invariant that $\ell=\Oh{\log n}$. This is where the three heaps in Sections~\ref{sec:llm-heaps}--\ref{sec:exponential} differ; each has its own invariants towards this goal, and each augments the basic implementations above with restructuring operations which maintain said invariants. All structures use the technique of combining sets in constant time by concatenating linked-lists and splitting sets in linear time using a linear time selection algorithm to maintain their invariants. Further, all structures use amortized analysis to argue for the running times.

At a very high level, the philosophy of each structure is at follows: 
In \emph{lazy partition heaps} (LP heaps) in \wref{sec:llm-heaps}, the sets are allowed to become arbitrarily large, and there is no lower bound on the number of sets. A bound on the size of a set compared to the sum of all previous sets is maintained.
In \emph{Fibonacci heaps: the next generation } (FH:TNG) in \wref{sec:Fibonacci} and the heaps in \wref{sec:exponential}, sets are allowed to be empty and there are strict upper bounds on their size. In FH:TNG, there are lower bounds on non-empty set sizes, where both the upper and lower bounds are given by the Fibonacci numbers, which naturally is amiable to concatenating and splitting while maintaining the size invariants. The heaps in \wref{sec:exponential} do not maintain a lower bound on the set sizes which results in simpler restructuring algorithms.

\section{Lazy Partition Heaps}
\label{sec:llm-heaps}

To keep the number of sets $\ell$ in $\Oh{\lg n}$ at all times,
\emph{\CSHeapsLong{}s} (\CSHeaps heaps) use the following pivot-forgetting rule, which is adapted from lazy search trees~\cite{SandlundWild20}:

\begin{quote}\textsl{A pivot $p$ shall be abandoned if the two sets that it separates together contain fewer elements than there are \emph{smaller} elements in the heap (\wref{fig:M}).}
\end{quote}

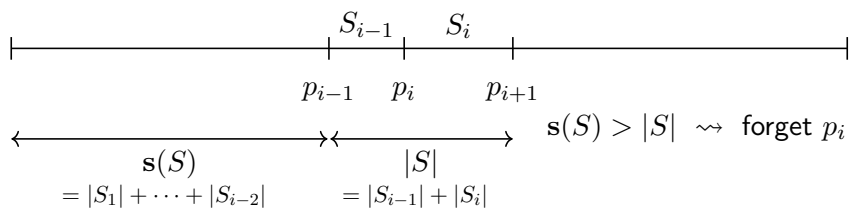
\begin{figure}[tbh]
\plaincenter{
	\begin{tikzpicture}[xscale=11,yscale=3,semithick]
	\sffamily
	
		\draw (0,0) -- (1,0) ;
		\foreach \x/\l in {0.38/{$p_{i-1}$},0.47/{$p_i$},0.6/{$p_{i+1}$},0/{},1/{}} {
			\draw (\x,.05) -- (\x,-0.05) node[below] {\strut\l} ;
		}
        \node at (0.425, 0.1) {$\Int{i-1}$};
        \node at (0.535, 0.1) {$\Int{i}$};
		\draw[<->,shorten >=.5pt] (0,-.4) -- node[below] 
			{$\out{}$} node[below=14pt,scale=.8]
			{${}\mathrel{\llap{=}} \lrSize{\Int{1}}+\cdots+\lrSize{\Int{i-2}}$} 
			++(0.38,0) ;
		\draw[<->,shorten <=.5pt] (.38,-.4) -- node[below] 
			{$\lrSize{\Int{}}$} node[below=14pt,scale=.8]
			{${}\mathrel{\llap{=}} \lrSize{\Int{i-1}}+\lrSize{\Int{i}}$} 
			++(0.22,0) ;
		\node at (0.82,-0.35) {$\out{} > \lrSize{\Int{}} $ \;$\leadsto$\; forget $p_i$} ;
	\end{tikzpicture}
}
\iflncs{\vspace*{-2ex}}{}
\caption{
	The pivot-forgetting rule for sets in \CSHeaps heaps.
	Here, the lists for sets $\Int{i-1}$ and $\Int{i}$ would be concatenated (thus forgetting $p_i$)
	since there are more elements smaller than $p_{i-1}$ than between $p_{i-1}$ and~$p_{i+1}$.
}
\label{fig:M}
\end{figure}

\noindent
In \CSHeaps heaps, every \DELETEMIN operation partitions $\Int{1}$ around its median.
Since no stringent bound on $\lrSize{\Int{1}}$ is guaranteed, we need to amortize the occasional high costs for \DELETEMIN.
For that, sets $\Int{j}$, $j=1,\ldots,\ell$, accumulate potential when they have grown ``too large'',
namely $\Pot{j} \coloneq \maxzero{\lrSize{\Int{j}}-\out{j}}$, where $\out{j} \coloneq \lrSize{\Int{1}}+\cdots+\lrSize{\Int{j-1}}$ denotes the number of elements smaller than all elements in $\Int{j}$ (their ``safety buffer'' from the minimum).

This in particular includes the 
set $\Int{1}$, whose potential $\Pot{1} = |\Int{1}|$ can thus pay for the scanning and partitioning of $S_1$ upon the next \DELETEMIN.

Apart from the sets, we store a (static) sorted array $T$ 
of the current sets $\Int{1},\Int{2},\ldots,\Int{\ell}$ (in this order).
(Whenever the collection of sets changes because of a {\DELETEMIN}, we can afford to rebuild~$T$ from scratch.)
For simplicity of the pseudocode, we let each set $\Int{i}$ store the corresponding pivot $p_i$ as $\attrib{\Int{i}}{p}$.
We point out that $\out{j}$ is only needed during the pivot-forgetting phase and is computed on demand there.

\subsection{Operations}

Pseudo code for the operations on a \CSHeaps heap is shown in \wref{fig:LP-heaps-pseudo-code}.
Throughout the lifetime of an \CSHeaps heap, we maintain the following invariant.

\begin{description}
\item[{Invariant (\Sinv):}]
	The number of set is at most $\ell \le 2\lg n + 1$. 
\end{description}

If the invariant is at risk of violation---which is only possible upon \DELETEMIN---we apply the following concatenation rule (cf.\ \wref{fig:M}).
Recall that $\out{j} \coloneq \lrSize{\Int{1}} + \cdots + \lrSize{\Int{j-1}}$ is the number of \emph{\underline{s}maller elements} (smaller than $p_{j}$, that is):

\begin{description}
\item[Concatenation Rule (C):]
	If \( \lrSize{\Int{j}} + \lrSize{\Int{j+1}} < \out{j}\), concatenate $\Int{j}$ and $\Int{j+1}$.
\end{description}

\begin{figure}

\begin{codebox}
	\Procname{\proc{\CSHeaps-Heap-Insert}($e$)}
	\li $\id{ptr}\gets \proc{New-Node}(e)$ \Comment obtain pointer to element
	\li $\Int{}\gets \attrib{T}{\proc{Find}}(\attrib{e}{key})$ \Comment binary search the pivots
    \li $\attrib{\Int{}}{\proc{Append}}(\id{ptr})$ \Comment insert element
	\li \Return $\id{ptr}$
    \setcounter{codelinenumber}{10}
\end{codebox}
\bigskip
\begin{codebox}
	\Procname{\proc{\CSHeaps-Heap-Decrease-Key}($\id{ptr}$, $\newkey$)}
    \li $e \gets {*}\id{ptr}$ \Comment obtain element from pointer
    \li $\Int{} \gets \attrib{T}{\proc{Find}}(\attrib{e}{key})$
    \li $\attrib{\Int{}}{\proc{Remove}(\id{ptr})}$ \Comment remove element from current set
        \li $\attrib{e}{key} \gets \newkey$ \Comment update key
    \li $\Int{} \gets \attrib{T}{\proc{Find}}(\attrib{e}{key})$
    \li $\attrib{\Int{}}{\proc{Append}}(\id{ptr})$ \Comment insert element in new set
    \setcounter{codelinenumber}{10}
\end{codebox}
\bigskip
\begin{codebox}
	\Procname{\proc{\CSHeaps-Heap-Delete-Min}()}
    \li $\id{min} \gets \attrib{\Int{1}}{\DELETEMIN()}$ \Comment scans $\Int{1}$
    \li $p \gets \proc{Median}(\Int{1})$ \Comment (larger) median (deterministic linear-time selection~\cite{BlumFloydPrattRivestTarjan1973})
    \li ${\IntAfter{1}}, {\IntAfter{2}} \gets \proc{Partition}(\Int{1}, p)$
    \li $\attrib{{\IntAfter{2}}}{p} \gets p$
	\li $\mathcal S \gets [{\IntAfter{1}},{\IntAfter{2}},\Int{2},\ldots,\Int{\ell}]$ \Comment doubly-linked list of sets
    \li $\proc{Forget-Pivots}(\mathcal S)$ \Comment modifies $\mathcal S$
    \li Rebuild $T$ from $\mathcal S$
	\li \Return $\id{min}$
    \setcounter{codelinenumber}{10}
\end{codebox}
\bigskip
\begin{codebox}
    \Procname{\proc{Forget-Pivots}($\mathcal S$)}
	\zi \Comment Assume $\mathcal S=[\Int{1},\ldots, \Int{\ell}]$
    \li $\attrib{\mathcal S}{\proc{Remove-All-Empty-Sets}}()$
    \li $A \gets \Int{1}$
    \li $\outBase{A} \gets 0$
	\li \While $\attrib{\mathcal S}{\proc{Has-Next}}(A)$
	\Do
        \li $B \gets \attrib{\mathcal S}{\proc{Next}}(A)$
		\li \If $\lrSize{A} + \lrSize{B} < \outBase{A}$
		\Then
			\li Concatenate $B$ into $A$
			\li $\attrib{\mathcal S}{\proc{Remove}}(B)$
        \li \Else
            \li $\outBase{B} \gets \outBase{A} + \lrSize{A}$
            \li $A \gets B$
		\EndIf
	\EndWhile \label{li:long-line}
\end{codebox}

\caption{\CSHeaps heaps implementation.}
\label{fig:LP-heaps-pseudo-code}
   
\end{figure}

Concatenation here simply means 
we concatenate the linked lists of elements for the two sets and ``forgets'' the pivot formerly separating the two sets. This operation takes constant time.
\proc{Forget-Pivots} computes $\out{j}$ and applies rule~(C) to $\mathcal S$ successively.
Clearly, \proc{Forget-Pivots} uses $\Oh{\ell}$ time if the (initial) number of sets is~$\ell$. 
Note that \proc{Forget-Pivots} also removes all empty sets present, thereby ensuring that $\lrSize{\Int{j}} \ge 1$ for all $1 \leq j \leq \ell$.
After \proc{Forget-Pivots} has concatenated all pairs of sets that satisfy~(C),
we obtain a bound on the number of sets from the following lemma.

\begin{lemma}[(C) implies (\Sinv)]
\label{lem:logarithmic_set_count}
    Suppose after applying rule (C), we have sets $\Int{1},\ldots, \Int{\ell}$, with $\lrSize{\Int{j}} \ge 1$ for all $1 \leq j \leq \ell$.
    Then $\ell \le 2 \lg n + 1$, for $n$ the total size of the sets.
\end{lemma}

\begin{proof}
    We first show that prefixes of sets grow exponentially: for all $1 \leq j \leq \ell$, we have
    $\Sigma(j)\coloneq \lrSize{\Int{1}}+\cdots+\lrSize{\Int{j}} \rel\ge 2^{(j-1)/2}$ $(*)$.
    
    The proof is by induction on $j$. The claims holds for $1 \leq j \leq 2$, since $\lrSize{\Int{i}} \ge 1$ for $1 \leq i \leq \ell$. In the inductive step for $j \ge 3$, the inductive hypothesis implies $\Sigma(j-2) \ge 2^{(j-3)/2}$.  
    By concatenation rule (C), we have $\lrSize{\Int{j-1}} + \lrSize{\Int{j}} \ge \out{j-1}$, and by definition $\out{j-1} = \Sigma(j-2)$. Therefore $\Sigma(j) = \Sigma(j-2) + \lrSize{\Int{j-1}} + \lrSize{\Int{j}} \ge 2 \cdot \Sigma(j-2) \ge 2 \cdot 2^{(j-3)/2} = 2^{(j-1)/2}$, which proves $(*)$.

    Since $(*)$ holds for any prefix of sets, it does so in particular for $j=\ell$. 
    Then, $\Sigma(\ell) = n$ (all elements), and we get
    $\ell \le 2 \lg n + 1$ as claimed.
\end{proof}

It follows from \wref{lem:logarithmic_set_count} that \proc{Forget-Pivots} ensures (\Sinv).
This completes the description of the data structure.

\subsection{Potential}

For analyzing the amortized cost of operations, 
we use the following potential:
\[ 
	\Phi \wrel\coloneq \sum_{j=1}^{\ell} \Pot{j} 
		\qquad\text{with}\qquad
	\Pot{j} \wrel\coloneq \beta \cdot \maxzero{ \lrSize{\Int{j}} - \out{j} } \; . 
\]
(The sum is over all sets currently in the data structure.)
Here $\beta\ge 2$ is a fixed scaling parameter (chosen depending on the constant factor in the linear cost of partitioning).

\subsection{Amortized Cost of Insert}

Upon an \INSERT operation,
the set $\Int{j}$ containing $e$ is located using binary search in $T$ in $\Oh{\lg\lg n}$ time.
Adding $e$ to the linked list of $\Int{j}$ takes constant time.

In $\Phi$, $\out{j'}$ increases by $1$ for all $j'>j$; this never increases the value of $\Phi$. 
The size of $\Int{j}$ grows by $1$ which increases $\Phi$ by at most~$\beta = \Oh{1}$.
\INSERT thus costs $\Oh{\lg \lg n}$ (both amortized and worst case).

\subsection{Amortized Cost of Decrease-Key}

When $\DECREASEKEY(\mathit{ptr}, \newkey)$ is performed, the element $e$ at pointer $\mathit{ptr}$ has its key decreased to $\newkey$. Re-inserting $e$ costs $\Oh{\lg \lg n}$ (both amortized and worst case).
This changes the size of at most two sets, with one increasing and the other decreasing, which increases the potential by at most $\beta$.
It may also change $\out{j'}$ for some sets $\Int{j'}$,
but since we move $e$ from $\Int{i}$ to an $\Int{j}$ with $1 \leq j \leq i$
when we make the key of the element smaller,
$\out{j}$ can only \emph{increase}, 
which does not increase the potential.
In total, \DECREASEKEY uses $\Oh{\lg \lg n + \beta} = \Oh{\lg \lg n} $ time (both amortized and worst case).

\DECREASEKEY may leave a set $\Int{i}$ empty, where $2\leq i \leq \ell$, however, these may be removed upon a later \proc{Forget-Pivots}, without an increase in time or potential, as the number of sets remains bounded, maintaining invariant (\Sinv).

\subsection{Amortized Cost of Delete-Min}

A \DELETEMIN scans $\Int{1}$ to find the minimum element $e$ and removes it;
moreover, we split $\Int{1}$ by partitioning it around its median,
resulting in two new sets ${\IntAfter{1}}$ and~${\IntAfter{2}}$.
For the partition, we use deterministic selection. 
Partitioning $\Int{1}$ has overall actual cost of $\Oh{\lrSize{\Int{1}}}$.

Finally, we establish invariant (\Sinv) by calling \proc{Forget-Pivots}.
By the invariant~(\Sinv) on the number of sets, and since partitioning set $\Int{1}$ creates two new sets, \proc{Forget-Pivots} uses amortized $\Oh[\big]{\lg n}$ actual time.

It remains to analyze the change in potential.
Note first that the concatenation rule (C) is chosen precisely so that when concatenating $\Int{j+1}$ into $\Int{j}$, the contributions $\Pot{j+1}$ and $\Pot{j}$, both before and after the concatenation, are all zero, so concatenating does not change the potential at all.
More specifically, assume that $\lrSize{\Int{j}}~+~\lrSize{\Int{j+1}}~<~\out{j}$, and a concatenation of $\Int{j}$ and $\Int{j + 1}$ is performed to obtain the new set ${\IntAfter{j}}$, by the concatenation rule (C). Then the potential is decreasing by $\Pot{j}$ and $\Pot{j+1}$ and increasing by the new potential ${\PotAfter{j}}$. By definition these potentials are
\[\begin{alignedat}{3}
    \Pot{j} 
    	&\wwrel= \beta \cdot \maxzero{ \lrSize{\Int{j}} - \out{j} } 
    	&&\wwrel= 0 
    \\
    \Pot{j+1}
    	&\wwrel= \beta \cdot \maxzero{ \lrSize{\Int{j+1}} - \lrParens{\out{j} + \lrSize{\Int{j}}} } 
    	&&\wwrel= 0 
    \\
    {\PotAfter{j}}
    	&\wwrel= \beta \cdot \maxzero{ \lrParens{ \lrSize{\Int{j}} + \lrSize{\Int{j+1}}} - \out{j} } 
    	&&\wwrel= 0 \; .
\end{alignedat}\]
The change in potential $\Delta \Phi$ is therefore $0$ upon a concatenation.

\begin{figure}[tbh]
	\centering
	\begin{tikzpicture}[every node/.style={font=\small},semithick,xscale=.77]
        \draw[|-|] (3.05, 0) -- node[below=0.55, anchor=south] {${\IntAfter{1}}$} (3.7, 0);
        \draw[|-|] (3.8, 0) -- node[below=0.55, anchor=south] {${\IntAfter{2}}$} (4.45, 0);
        \draw[|-|] (4.55, 0)  -- node[below=0.55, anchor=south] {$\Int{2}$} (5.6, 0);
        \draw[|-|] (5.7, 0)   -- node[below=0.55, anchor=south] {$\Int{3}$} (8.05, 0);
        \draw[|-|] (8.15, 0)  -- node[below=0.55, anchor=south] {$\Int{4}$} (12.95, 0);
        \draw[|-|] (13.05, 0) -- node[below=0.55, anchor=south] {$\Int{5}$} (18, 0);
        
        \draw[decorate,decoration={brace,amplitude=8pt}] (3.05, 0.25) -- (4.45, 0.25) node[above, midway, yshift=7pt] {\footnotesize $\Int{1}$};
        
        \draw[dashed] (3, 0.2) -- (3, -.5) node[pos=1, fill=white, yshift=-3.2pt,inner sep=2pt] {min};
	\end{tikzpicture}
	\iflncs{\vspace*{-2ex}}{}
    
	\caption{Sketch of the sets of an \CSHeaps heap and the partition of $\Int{1}$ into ${\IntAfter{1}}$ and ${\IntAfter{2}}$ in \DELETEMIN.}
	\label{fig:partition}
\end{figure}

It remains to analyze the influence on the potential of partitioning the 
set $\Int{1}$ into the new sets ${\IntAfter{1}}$ and ${\IntAfter{2}}$.
The situation after partitioning $\Int{1}$ is illustrated in \wref{fig:partition}.
For the set $\Int{1}$, we lose $\Pot{1} = \beta \lrSize{\Int 1}$ and gain ${\PotAfter{1}} + {\PotAfter{2}} = \beta\lrSize{{\IntAfter{1}}}$ from ${\IntAfter{1}}$ and~${\IntAfter{2}}$, as we choose the larger median s.\,t.\ $\mathbf{s}\lrParens{{\IntAfter{2}}} = \lrSize{{\IntAfter{1}}} \ge \lrSize{{\IntAfter{2}}}$, i.e., ${\PotAfter{2}}=0$.
All other sets lose one smaller element, so we gain $\le \ell\cdot \beta$ in potential.
By construction, $\lrSize{{\IntAfter{1}}} \le \frac12 \lrSize{\Int{1}}$ (recall that we delete an element from $\Int{1}$ before partitioning into ${\IntAfter{1}}$ and ${\IntAfter{2}}$),
so $\Delta\Phi \le - \frac\beta2 \lrSize{\Int{1}} + \ell \beta$.
For a sufficiently large $\beta \ge 2$, the total amortized cost is thus
\[
	\Oh{\Delta\Phi + \lrSize{\Int{1}} + \lg n}
	\wwrel=
	\Oh{\lg n} \; .
\]

\subsection{Extensions}

Here we discuss a few further properties of \CSHeaps heaps;
they are not needed to understand (or teach) \CSHeaps heaps.

\subsubsection{Further Operations}

We can support \proc{\CSHeaps-Heap-Find-Min} in $\Oh{1}$ time, by maintaining the minimum upon \proc{\CSHeaps-Heap-\INSERT} and \proc{\CSHeaps-Heap-\DECREASEKEY} in constant time, and by scanning ${\IntAfter{1}}$ after a \proc{\CSHeaps-Heap-\DELETEMIN}, which maintains the amortized time.
Note that $|S_1|$ can be arbitary large, i.e., we cannot afford scanning $S_1$ for each \proc{\CSHeaps-Heap-Find-Min} operation.

We can also support an \proc{\CSHeaps-Heap-\emph{Increase}\/-Key} operation or a general \proc{\CSHeaps-Heap-Delete}(\id{ptr}) operation, at total amortized cost $\Oh{\beta \lg n + \lg \lg n} = \Oh{\lg n}$.
For both operations $S_1$ might become empty, and further 
for general \proc{\CSHeaps-Heap-Delete}(\id{ptr}), $n$ decreases which may violate invariant (\Sinv), and therefore \proc{Forget-Pivots} is needed to re-establish (\Sinv). This also adds $\Oh{\lg n}$ to the cost.

An operation \proc{\CSHeaps-Heap-Build} from a given collection of $n$ elements can be supported in $\Oh{n}$ time.
For that, we simply declare all elements in $\Int{1}$.
The actual cost (if elements are given in a linked list) can even be constant, however we have to pay
for the increase in potential of $\beta n$, so total amortized costs are $\Oh{n}$.

The \proc{\CSHeaps-Heap-Meld} operation of taking the union of two heaps is not easily supported efficiently.
One can trivially create a new \CSHeaps heap from the concatenation of the two input heaps using \proc{\CSHeaps-Heap-Build}, at linear amortized cost in the total number of elements.

\subsubsection{Partitioning}

We can replace the deterministic linear-time selection~\cite{BlumFloydPrattRivestTarjan1973} by randomized selection~\cite{Hoare1961,FloydRivest1975}, leading to the same expected amortized time. 

More importantly, we can also simply run a \emph{single round} of (random) partitioning:
Upon the partition, the potential decreases by $\min \left\{ \lrSize{{\IntAfter{1}}} ,\; \lrSize{{\IntAfter{2}}} \right\}$. By choosing the pivot not as the median, but at random between the elements of~$\Int{1}$, the decrease is therefore \emph{expected} linear in $\lrSize{\Int{1}}$, and the amortized time of \DELETEMIN remains \emph{expected} $\Oh{\lg n}$.

\subsubsection{\CSHeaps Quickheaps}

The \emph{quickheaps layout} (described in more detail in \wref{sec:related-work}) offers an enticing alternative to the linked lists for each set: 
At the expense of more data movement during insertions and deletions, 
we store all sets contiguously in a single array, with the pivots separating them.
The only extra space is a static array $T$ with the $O(\lg n)$ pivot positions.

Concatenations of adjacent sets can still be done in constant time 
(by simply demoting a pivot to an ordinary element),
and partitioning can now be done in place, which improves practical overhead in running time. 
While we can also still binary search the pivots to \emph{find} a set to insert an element into,
insertions or deletions must now be implemented with up to one swap per set, leading to overall $\Theta(\lg n)$ running time for all operations.

\subsubsection{Pointer-Machine Version}

\CSHeaps heaps already stay within the pointer-machine model of computation with bounded indegree nodes except for the index $T$.
By replacing the sorted array and binary search by any balanced binary search tree,
then \CSHeaps heaps is entirely in the model.

\section{Fibonacci Heaps: The Next Generation}
\label{sec:Fibonacci}

\DeclarePairedDelimiter{\nint}\lfloor\rceil

\newcommand{\jafter}[1]{\after{#1}}
\newcommand{\jbefore}[1]{#1}

\newcommand{\downsum}[1]{\uparrow_{#1}}
\newcommand{\potdown}[1]{\phi^{\uparrow}_{#1}}
\newcommand{\potsize}[1]{\phi^{\text{size}}_{#1}}
\newcommand{\potset}[1]{\phi^{\text{nonempty}}_{#1}}
\newcommand{\pot}[1]{\phi_{#1}}
\newcommand{\downsuma}[1]{\jafter{\uparrow}_{#1}}
\newcommand{\potdowna}[1]{\jafter{\phi}^{\uparrow}_{#1}}
\newcommand{\potsizea}[1]{\jafter{\phi}^{\text{size}}_{#1}}
\newcommand{\potseta}[1]{\jafter{\phi}^{\text{nonempty}}_{#1}}
\newcommand{\pota}[1]{\jbefore{\phi}_{#1}}
\newcommand{\downsumb}[1]{\jbefore{\uparrow}_{#1}}
\newcommand{\potdownb}[1]{\jbefore{\phi}^{\uparrow}_{#1}}
\newcommand{\potsizeb}[1]{\jbefore{\phi}^{\text{size}}_{#1}}
\newcommand{\potsetb}[1]{\jbefore{}{\phi}^{\text{nonempty}}_{#1}}
\newcommand{\potb}[1]{\jbefore{\phi}_{#1}}

\newcommand{\figba}[4]{
\begin{figure}[tbp]
\begin{center}
 \includegraphics[width=\linewidth]{#1}\\
 {\Huge $\downarrow$}\\
 \includegraphics[width=\linewidth]{#2}
\end{center}
\caption{#3}
\label{#4}
\end{figure}
}

The next generation Fibonacci heap (FH:TNG)\footnote{The name is meant to reflect that fact that in the original Fibonacci heaps and the the structure here, Fibonacci numbers play crucial, and entirely different, roles.} follows our framework for partition-based simple heaps with but a few distinctive features: Sets can either be present or absent. If a set is present, its size is required to be between the $i$th Fibonacci number, denoted as $F_i$ ($F_1\coloneqq F_2 \coloneqq 1$; $F_i\coloneqq F_{i-1}+F_{i-2}$), and the $i+3$rd, $F_{i+3}$, with the one exception that the first set, $S_3$, has no minimum size. Finally we require that there are at most eight empty sets in a row (\emph{the consecutive empty} invariant) and at most two non-empty sets in a row (\emph{the consecutive nonempty} invariant) to allow room for splits and concatenations.%

\begin{figure}
\includegraphics[width=\linewidth]{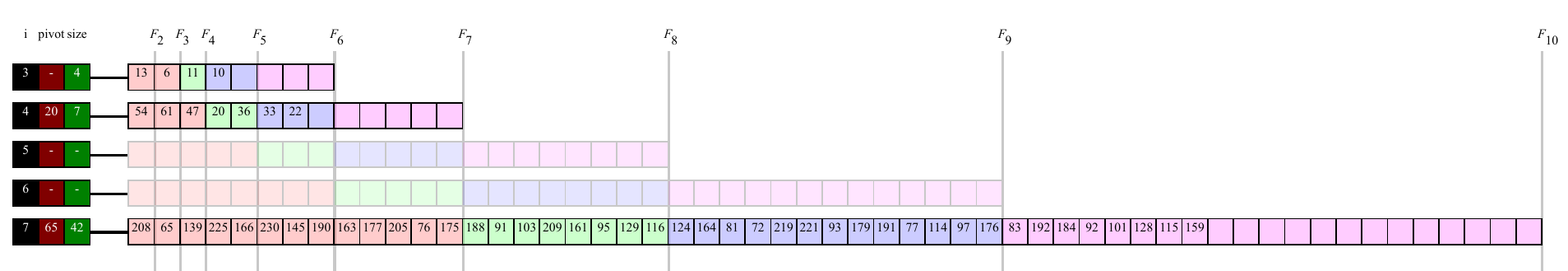}
\caption{Illustration of a next generation Fibonacci heap. While each set is stored as a linked list, this illustration uses colors to show the relationship between each list's size, the Fibonacci numbers, and the size potential function.}
\label{fig:fhtng}
\end{figure}

The implementation of the operations follows our basic framework. The only difference is what to do when the invariants are violated.
We restore them as follows using \emph{invariant-restoring operations}:
If a set $S_i$ has $F_{i+3}$ elements we call it \emph{full} and there are two options: If the set below, $S_{i+1}$ is empty, $S_i$ is simply moved to $S_{i+1}$. We call this a \emph{overflow-down operation}.
If $S_{i+1}$ is not empty, $S_{i+2}$ must be empty by the consecutive nonempty invariant. We concatenate all $F_{i+3}$ elements of $S_i$ to the set $S_{i+1}$ and then remove the $F_{i+3}$ largest elements of the resultant set using linear-time selection and place them in $S_{i+2}$; this leaves the size of $S_{i+1}$ unchanged but the contents will certainly change. We call this an \emph{overflow-thru operation}.
If a set $S_i$ has reached its minimum size $F_i$, it is \emph{underfull}, and we execute either an \emph{underflow-up operation}  or  an \emph{underflow-thru operation} depending on whether $S_{i-1}$ is empty. These are symmetric to the two overflow operations.
If the consecutive-nonempty invariant is violated, there are three nonempty sets in a row with an empty list below. We concatenate the top two of them and place the resultant set in the empty slot below, we call this a \emph{merge-down operation}. 

If the consecutive-empty invariant is violated, there are nine empty sets in a row with a non-empty set below. We take the nonempty set below and split using linear-time selection into two set slots immediately above. We call this a \emph{split-up}, and this takes linear time in the sets involved.
Merge-down takes constant time as it is combining two linked lists, whereas the split-up takes linear time in the size of the sets to perform a selection and traverse the set to split it in two.

The asymmetry of the linear-time split-up with the constant time merge-down is crucial and is deeply connected to the fact that we allow keys to be decreased but not increased! The intuition is that repeated {\DECREASEKEY} operations  can force inexpensive merge-downs to be performed as everyone is slowly moved to bigger sets. For example, this can happen by repeatedly \DECREASEKEY-ing the maximum item to make it the minimum.
 However, split-up operations are only triggered by the more expensive {\DELETEMIN} operations and thus their more expensive cost can be paid for.
 
\ARXIV{Here we present the potential function used in the analysis.}
\LATIN{Here we present the potential function used in the analysis, leaving the straightforward but detailed proofs to the full version.} 
The potential of the data structure is the sum of three components, each of which is defined for each set in the structure:

\begin{description}
\item[Nonempty potential.] 
Define $\potset{i}$ to be 1 if $S_i$ is nonempty and 0 if $S_i$ is empty.  This is used to pay for the merge-down operation where the number of nonempty sets decreases by one at unit cost.

\item[Size potential.] 
Gives potential as $S_i$ nears being full ($F_{i+3}$ elements) or underfull ($F_i$ elements) and 0 potential when its size is in the middle of its allowable range (from $F_{i+1}$ to $F_{i+2}$ elements).  This is the classic potential function used in array doubling and is used to pay for the underflow and overflow operations. Figure~\ref{fig:fhtng} is color-coded to highlight these three ranges.
    \[ \potsize{i}\coloneqq
    \begin{cases}
    0 & \text{if $S_i$ is empty}\\
        F_{i+1}-|S_{i}| & \text{if } F_{i} \leq |S_i| < F_{i+1}  
        \text{ (green in the figure)}
\\        0 & \text{if } F_{i+1} \leq |S_i|\leq F_{i+2}
        \text{ (blue in the figure)}
\\
        |S_i|-F_{i+2} & \text{if } F_{i+2} <|S_i| \leq F_{i+3}    
                \text{ (purple in the figure)}
 \end{cases}\]

\item[Up potential.]
This gives potential to nonempty sets that have few items above. This is used to pay for the split-up operation. The intuition is that this is a mechanism that allows the \DELETEMIN\ operation to pay unit cost to each of $\Theta(\lg n)$ sets, which the sets save up in the case most items above have been removed and a split-up needs to happen.
Let $\downsum{i}$ be defined to be $\sum_{j=1}^{i-1}|S_j|$. Define the up potential of set $S_i$, $\potdown{i}$, as

\[ \potdown{i} \coloneqq 
\begin{cases}
0 & \text{if $i\leq 2$ or $S_i$ is empty}\\
    \maxzero{F_{i-3}-\downsum{i}}
& \text{otherwise}
\end{cases} \]
\end{description}

\ARXIV{ %
The potential of the data structure, $\Phi$, is the sum of the up, size, and set potentials of every set; we refer to this sum as $\pot{i}$: $\Phi \coloneqq \sum_{i=0}^k (\overbrace{\potdown{i}+\potsize{i}+\potset{i}}^{\pot{i}})$.

\paragraph{Facts} 

We now present a series of facts which relating to these potentials.

\begin{fact}\label{f1}
If $S_i$ is nonempty and at least of one of $S_{i-1}$, $S_{i-2}$ $S_{i-3}$ is nonempty then $\potdown{i}=0$.
\end{fact}

\begin{proof}
	If any of $S_{i-1}$, $S_{i-2}$ and $S_{i-3}$ are nonempty they will have least $F_{i-3}$ 
	items. Thus  $\downsum{i}\geq F_{i-3}$ and thus $\maxzero{F_{i-3}-\downsum{i}} =0$.
\end{proof}

\begin{fact}\label{f2}
	$\downsum{i} \leq F_{i+4}-1$
\end{fact} 

\begin{proof}
	$\downsum{i}=\sum_{j=1}^{i-1}|S_j| \leq \sum_{j=1}^{i-1} F_{i+3} 
	= \sum_{j=4}^{i+2} F_{i} 
	\leq  F_{i+4}-1 $ 
	via the classic identity that $\sum_{i=1}^{n} F_i= F_{i+2}-1$
\end{proof}

\begin{fact}\label{f3}
If an operation redistributes items among $S_{i}, S_{i+1}, \ldots S_j$ without insertion or deletion,  	$\potdown{\ell}$ does not change outside of $\ell \in [i,j]$.
\end{fact}

\begin{proof}
Recall that	$\potdown{\ell}$ is a function of whether $S_\ell$ is empty as well as $\downsum{\ell}$, which is defined to be $\sum_{m=1}^{\ell-1}|S_m|$ and does not change if the sets being redistributed are either entirely inside or entirely outside of this sum.
\end{proof}

\begin{fact}\label{f6}
For $i=O(1)$, $\phi_i=O(1)$.
\end{fact} 
The next fact, follows from the intuition that set $S_i$ can be split into $S_{i-1}	$ and $S_{i-2}$ or an inverse concatenation performed without gaining size potential. This follows from the properties of Fibonacci numbers, where a set of a certain fullness can be split into two with the same fullness and thus the same potential.

\begin{fact}\label{f7}
The elements of a nonempty set $S_i$ can be distributed to previously empty sets $S_{i-1}	$ and $S_{i-2}$ so that the set potential remains unchanged. Inversely, the elements of nonempty sets $S_{i-1}$ and $S_{i-2}$ can be concatenated into the previously empty set $S_i$ without increasing potential.
\end{fact}

\begin{proof}
In general, if $S_i$ has between $F_{i+j}$ and $F_{i+j+1}$ elements, for $j\in {0,1,2}$ we can proportionally distribute them so that $S_{i-1}$ has between $F_{i+j-1}$ and $F_{i+j}$ elements and 
 $S_{i-2}$ has between $F_{i+j-2}$ and $F_{i+j-1}$ elements; this follows naturally from the definition of Fibonacci numbers. In any such distribution, for the case where $j=0$ the size potential was and remains, $F_{i-1}-|S_i|$, for $j=1$, 0, and for $j=2$, $|S_i|-F_{i+2}$. 

For the inverse operation of merging, set $j$ so that $S_i$ will have between $F_{i+j}$ and $F_{i+j+1}$ elements. If $j=1$ then this holds trivially as the potential will become zero. For the cases where $j=0$ and $j=2$ a little math confirms the intuition:

\[
\underbrace{\potsize{i}}_{\text{after}} = |S_{i}|-F_{i+1}
 = 
\overbrace{|S_{i-1}|-F_{i}+|S_{i-2}|-F_{i-1} \leq \underbrace{\potsize{i-1}+\potsize{i-2}}_{\text{before}}}^
{\substack{|S_{i-1}|-F_{i}\leq \potsize{i-1}\\|S_{i-2}|-F_{i-1}\leq \potsize{i-2}}}
\]
\[
\underbrace{\potsize{i}}_{\text{after}}  = F_{i+3}-|S_{i}|
 = 
\overbrace{F_{i+2}-|S_{i-1}|+F_{i+1}-|S_{i-2}| \leq \underbrace{\potsize{i-1}+\potsize{i-2}}_{\text{before}}}^
{\substack{F_{i+2}-|S_{i-1}|\leq \potsize{i-1}\\F_{i+1}-|S_{i-2}|\leq \potsize{i-2}}}
\]
\end{proof}

	\begin{sidewaysfigure}
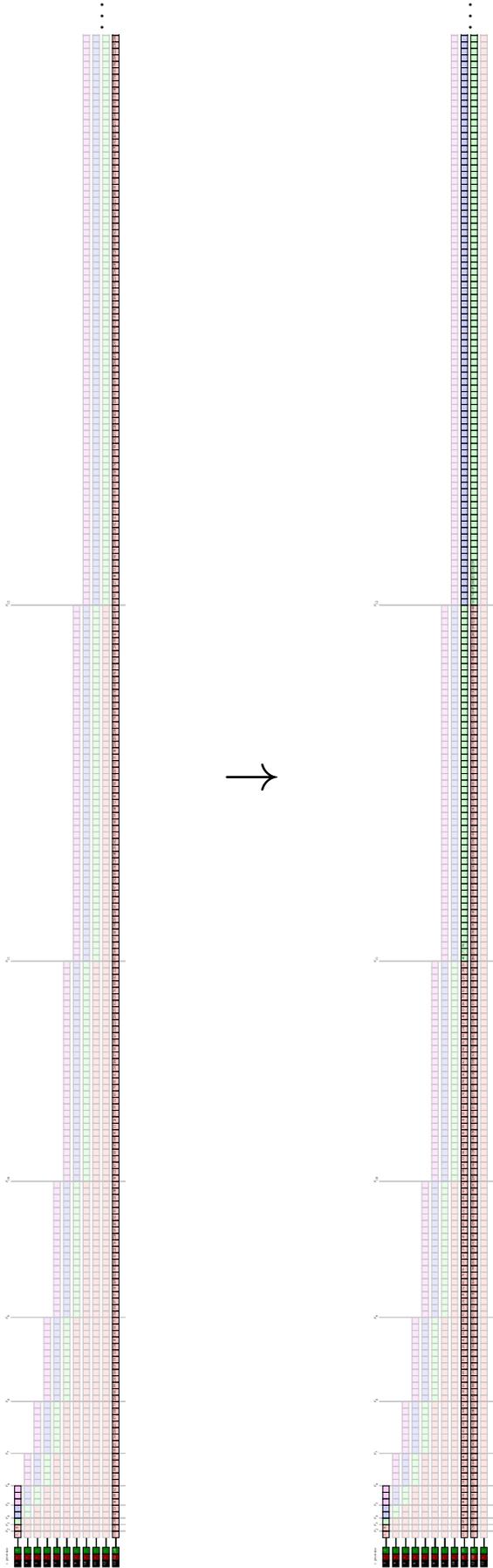

\begin{center}
 \hspace*{-1em}\includegraphics[trim=0cm 0cm 400cm 0cm,clip,width=\linewidth]{figs/su1.pdf}\,\raisebox{2ex}{\dots}\\
\ \vspace{1cm}\ 
 \\
 {\Huge $\downarrow$}\\
\ \vspace{1cm}\ 
 \\
 \hspace*{-1em}\includegraphics[trim=0cm 0cm 400cm 0cm,clip,width=\linewidth]{figs/su2.pdf}\,\raisebox{2ex}{\dots}
\end{center}
\caption{Split-up operation. There are 9 consecutive empty sets in a row, $S_2 \ldots S_{10}$. The set $S_{11}$ is then proportionally split into sets $S_{9}$ and $S_{10}$. In all cases the nonempty sets remain in the same color range, in this case green, which represents one of the three cases of the size potential. Empty cells continue off the page.}\label{f:su}
\end{sidewaysfigure}

\subsection{Split-up} \label{ss:sd}

See Figure~\ref{f:su}.
Split-up is called when the invariant that there are at most 8 consecutive empty sets in a row is violated, that is, there are 9 consecutive empty sets. Specifically, it is called on set $i$ when $S_i$ and $S_{i-10}$ are nonempty and $S_{i-1} \ldots S_{i-9}$ are empty. This requires $i\geq 11$.

The moving of items from $S_i$ to $S_{i-1}$ and $S_{i-2}$ should be done in a way proportional their sizes, that is in respect of the golden ratio. Very specifically, if there are between $F_l$ and $F_{l-1}$ elements in $S_i$, these should be moved so there are between $F_{l-1}$ and $F_{l-2}$ elements in $S_{i-1}$ and $F_{l-2}$ and $F_{l-3}$ elements in $S_{i-3}$; this is of course possible due to the definition of the Fibonacci numbers.

The new pivots are set to the minimum in each set via linear scan.

\begin{lemma}
	The amortized cost of the split-up operation is at most 0 when called on $S_i$, $i\geq 11$.
\end{lemma}
\begin{proof} 
The nominal cost of the split-up operation is defined to be $F_{i-6}$ as it takes linear time in the sets involved. To compute the change in potential we look at each of the three components of the potential function. 

\begin{description}
\item[Change in $\potdown{i}$:]
This is $\sum_{i=0}^k \potdownb{i}-\sum_{i=0}^k \potdowna{i}$. However, this can be simplified using Fact~\ref{f3} and recalling that empty sets have zero up potential to
$ \potdowna{i-1}+\potdowna{i-2}-\potdown{i}$.
We can now bound each of these quantities separately. First,
$\potdowna{i-1}=0$ by Fact~\ref{f1}.
Next, for $\potdowna{i-2}$, this is $\maxzero{F_{i-5}-\downsuma{i-2}} \leq F_{i-5}$.

Finally, for $-\potdownb{i}$: we use the fact that ${S}_{i-1}\ldots {S}_{i-9}$ are empty to bound $\downsum{i}$:
\[\begin{alignedat}{3}
    -\potdownb{i}
    &=-\maxzero{ F_{i-3}-\downsumb{i} }
    &\text{By definition}
    \\
    &=-\maxzero{ F_{i-3}-\downsumb{i-9} }
    &\text{$\downsumb{i}=\downsumb{i-9}$}
    \\
    & < -\maxzero{ F_{i-3}-F_{i-5}-1 }
    & \text{\ $\downsumb{i-9}<F_{i-5}-1$ by Fact~\ref{f2}}
    \\
    & = -F_{i-4}+1 &i \geq 11 
\end{alignedat}\]
 Thus, bringing it together gives:
 \[ \sum_{i=0}^k \potdownb{i}-\sum_{i=0}^k \potdowna{i}
= \potdowna{i-1}+\potdowna{i-2}-\potdownb{i}\leq 
0 +F_{i-5}-F_{i-4}+1=-F_{i-6}+1\]
    \item[Change in $\potsize{i}$:] 0. 
    \item[Change in $\potset{i}$:] +1, as there is now one more nonempty set.
\end{description} 
In summary,
by definition, the amortized cost is the nominal cost plus the change in potential which is at most, when $i\geq 11$: 
$
F_{i-6}-F_{i-6}-1+0+1 = 0 $
\end{proof}

\subsection{Merge-down} \label{ss:mu}

 \figba{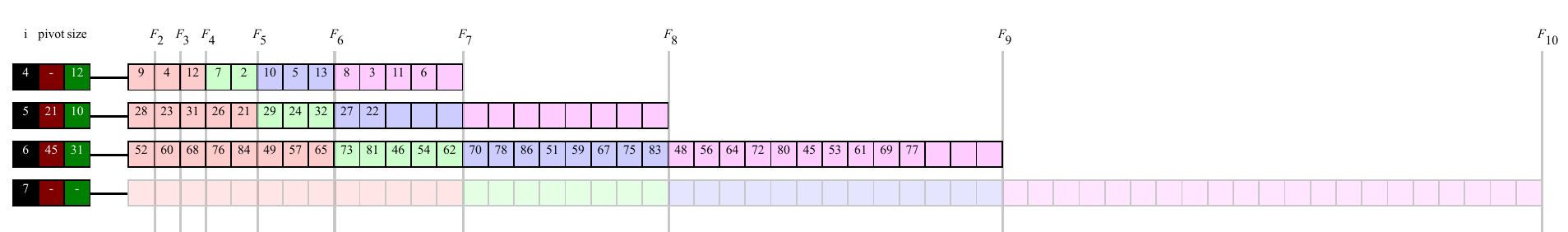}{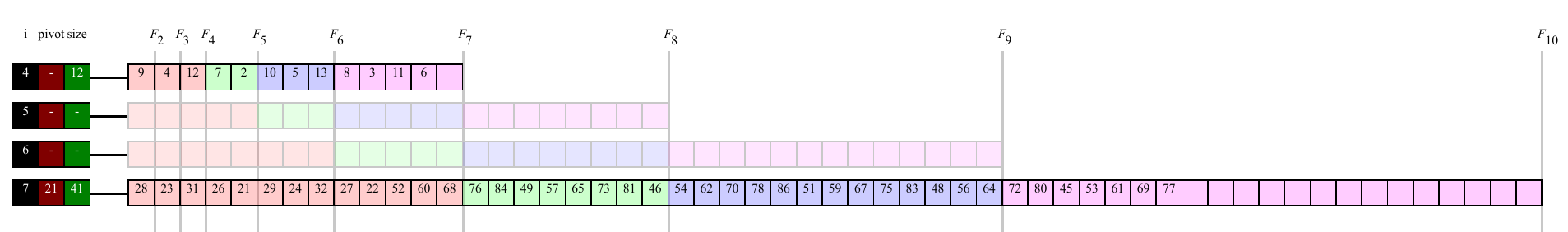}{A merge-down operation performed when there are three non-empty sets in a row, here $S_4, S_5, S_6$, and the set below $S_7$ is empty. sets $S_5$ and $S_6$ are concatenated into the empty slot of set $S_7$.}{f:md}

See Figure~\ref{f:md}. The merge-down operation is called when the invariant that there are most two non-empty consecutive sets is violated. That is, there are three non-empty consecutive sets $S_i$, $S_{i-1}$ and $S_{i-2}$ with $S_{i+1}$ being empty (or non-existent).
The sets $S_i$ and $S_{i-1}$ are concatenated to form the new set $\jafter{S}_{i+1}$; this is where we use fact that the sets are stored as doubly-linked sets and not as arrays so that this operation takes constant time.
The pivot $\overline{p}_{i-2}$ is set to $p_{i-1}$.

\begin{lemma}
	The amortized cost of the merge-down operation is at most 0.
\end{lemma}
\begin{proof}
The nominal cost of the operation is 1,  To compute the change in potential we look at each of the three components of the potential function: 
    \begin{description}
\item{\bfseries Change in $\potdown{i}$:} 0.
By Fact~\ref{f3} this is $\potdowna{i+1}-\potdownb{i}-\potdownb{i-1}$.
By Fact~\ref{f1}, since $\jbefore{S}_{i-2}=\jafter{S}_{i-2}$ is nonempty, all three of these are zero. 
    \item{\bfseries Change of $\potsize{i}$:} $\leq 0$ by Fact~\ref{f7}.
    \item{\bfseries Change of $\potset{i}$:} $-1$ as the number of non-empty sets decreases by one.
\end{description}
In summary, by definition, the amortized cost is the nominal cost plus the change in potential: 
$1+0+0-1 =0 $
\end{proof}

\subsection{Overflow-down} \label{ss:ou}

 \figba{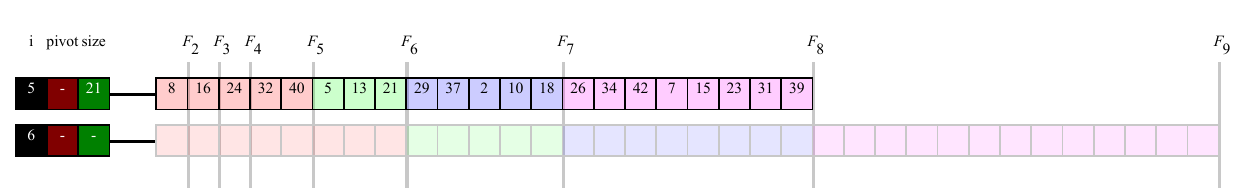}{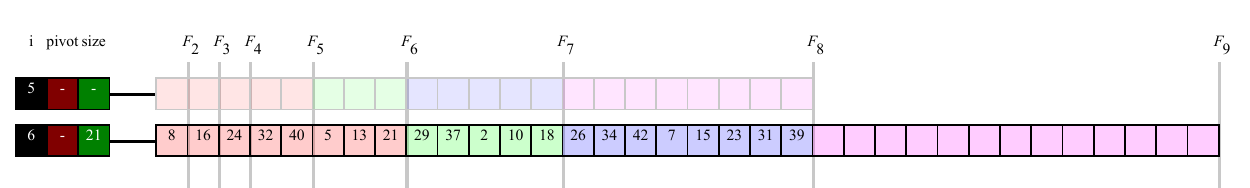}{An overflow-down operation performed on $S_5$ which has $F_8=21$ elements when $S_6$ is empty.}{f:od}

See Figure~\ref{f:od}. The overflow-down operation is called when a set $S_i$ is full, and thus has $F_{i+3}$ items, and the set below, $S_{i+1}$, is empty. The set is moved from $S_i$ to $S_{i+1}$ in constant time.

\begin{lemma}
	The amortized cost of the overflow-down on set $S_i$ is at most 0 when $i\geq 3$
\end{lemma}

\begin{proof}
The nominal cost of overflow up is 1, since it is a constant-time operation.
To compute the change in potential we look at each of the three components of the potential function: 

    \begin{description}
\item[Change in $\potdown{i}$:]
By Fact~\ref{f3} this is $\potdowna{i+1}-\potdownb{i}$ which is at most $\potdowna{i+1}$.
Using the definition of the up potential, 
$\potdowna{i+1}=\maxzero{ F_{i-2}-\downsumb{i+1} } \leq F_{i-2}$.
    \item[Change of $\potsize{}$:]
    As $|\jafter{S}_{i+1}|=F_{i+3}$, $\potsizea{i+1}=0$. Thus we need only consider $-\potsizeb{i}$:
$
 -(\overbrace{|\jbefore{S}_i|}^{F_{i+3}}-F_{i+2})
   = F_{i+2}-F_{i+3} 
    = -F_{i+1} $
    \item[Change of $\potset{i}$:] 0 as the number of nonempty sets does not change.
\end{description}
In summary, by definition, the amortized cost is the nominal cost plus the change in potential is 
$ 1+\overbrace{F_{i-2}-F_{i+1}}^{\leq -F_i}+0 \leq 0 $.
\end{proof}

\subsection{Overflow-thru}\label{ss:ot}

\figba{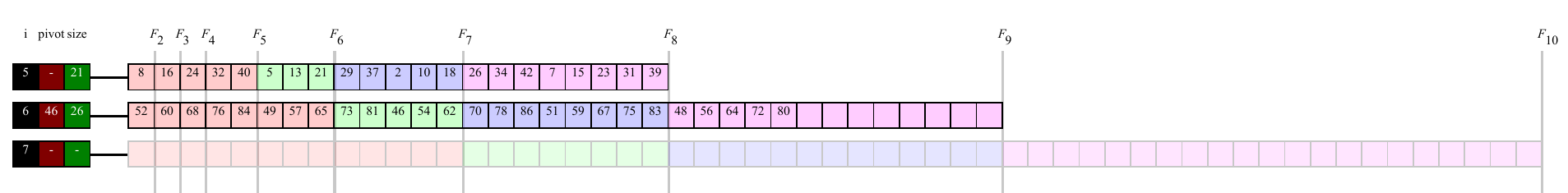}{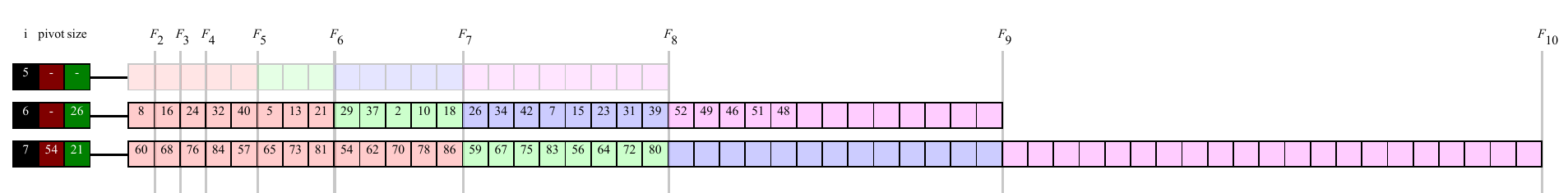}{An overflow-thru operation performed on $S_5$ which has $F_8=21$ elements when $S_6$ is nonempty.}{f:ot}

See Figure~\ref{f:ot}. The overflow-thru operation is called when a set $\jbefore{S}_i$ is full, and thus has $F_{i+3}$ items, and the set below, $\jbefore{S}_{i+1}$ is nonempty. 
As there can not be more than two non-empty consecutive sets by the consecutive nonempty invariant, this implies $\jbefore{S}_{i+2}$ is empty.

Recall that in this operation, $\jbefore{S}_i$ and $\jbefore{S}_{i+1}$ are combined, and the largest items are moved from the combined set to $\jafter{S}_{i+2}$ so that the size of $\jafter{S}_{i+1}$ is unchanged and thus $|\jafter{S}_{i+2}|=|\jbefore{S}_{i}|=F_{i+3}$. This movement requires a linear time median/order statistic algorithm and thus has nominal cost $F_{i-4}$ since $S_{i+2}$ is the largest set involved. 
The pivot $\jafter{p}_{i+1}$ and $\jafter{p}_{i+2}$ are set to the minimum in each set.

\begin{lemma}
	The amortized cost of overflow-thru is $\leq 0$, when $i\geq 5$.
\end{lemma}

\begin{proof}
The nominal cost of the operation is $F_{i-4}$ as it takes linear time in the size of the sets involved.	To compute the change in potential we look at each of the three components of the potential function: 

\begin{description}
    \item[Change in $\potdown{}$:] 
    By Fact~\ref{f3} this is
    $\potdowna{i+2}+\potdowna{i+1}-\potdownb{i+1}-\potdownb{i}$.
    By Fact~\ref{f1}, $\potdownb{i+1}$ and $\potdowna{i+2}$ are both zero. Combined with the fact that $\potdownb{i+1}$ is nonnegative,
     the change in $\potdownb{}$ is at most  
$    \potdowna{i+1}=\maxzero{ F_{i-2}-\downsumb{i+1} } \leq F_{i-2}$.
    \item[Change in $\potsize{}$:] 
    This is $\potsizea{i+2}+\potsizea{i+1}-\potsizeb{i+1}-\potsizeb{i}$.
   As $|S_{i+2}|=F_{i+3}$, $\potdowna{i+2}= 0$ by definition. Since 
    $\potsizeb{i+1}=\potsizea{i+1}$, this means the change in $\potsizeb{}$ is $-\potsizeb{i}$. By definition 
    $-\potsizeb{i}=-(|\jbefore{S}_i|-F_{i+2})$, which since $|\jbefore{S}_i|=F_{i+3}$ is $-F_{i+1}$.
    \item[Change in $\potsetb{}$: ] 0, since the total number of nonempty sets does not change.
\end{description}
In summary, by definition, the amortized cost is the nominal cost plus the change in potential: 
$ F_{i-4}
+ F_{i-2}
-F_{i+1}
+0
 \leq 0 $
\end{proof}

\subsection{Underflow-up}\label{ss:ud}

\figba{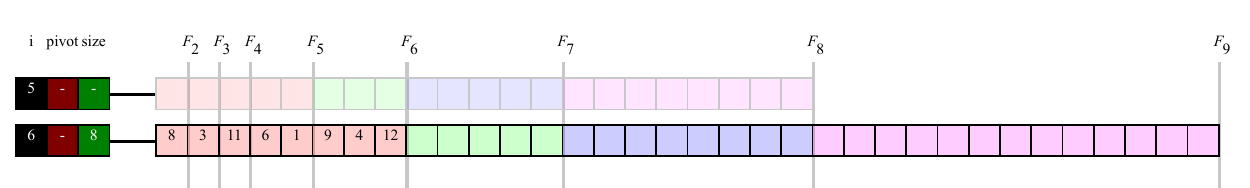}{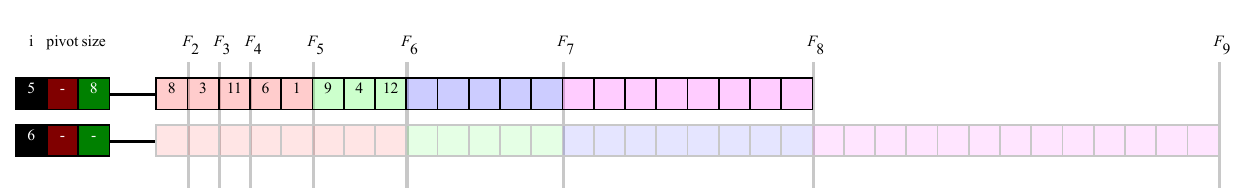}{An underflow-up operation performed on $S_6$ which has $F_6=8$ elements when $S_5$ is empty.}{f:uu}

See Figure~\ref{f:uu}. The underflow-up operation is called when a set $S_i$ is underfull, and thus has $F_{i}$ items, and the set above, $S_{i-1}$, is empty. In this case the set and the pivot is moved from $S_i$ to $S_{i-1}$ in constant time.

\begin{lemma}
	The amortized cost of the underflow-up operation on $S_i$ is at most 0 when $i\geq 6$ .
\end{lemma}

\begin{proof}
	The nominal cost of the operation is 1, as it is a constant-time operation.
		To compute the change in potential we look at each of the three components of the potential function: 

\begin{description}
    \item[Change in $\potdown{}$:] 
    By Fact~\ref{f3} this is
    $\potdowna{i-1}-\potdown{i}$. We have $ \potdowna{i-1}-\potdown{i} \leq \potdowna{i-1} \leq F_{i-4}$ since $\potdowna{i-1}=\maxzero{ F_{i-4}-\downsuma{i-1} }$ by definition.
    
    \item[Change in $\potsize{}$:] This is $   	\potsizea{i-1}-\potsize{i}$. As $|\jafter{S}_{i-1}|=F_i$, $\potsizea{i-1}=0$ by definition. We can now bound $    -\potsize{i}$:
$    	 - (F_{i+1}-\overbrace{|\jbefore{S}_i|}^{F_i})
    \leq -	F_{i+1}+F_i
    = -F_{i-1}    
    $.
    \item[Change in $\potset{}$: $0$.] The number of nonempty sets does not change. 
\end{description}

In summary, by definition, the amortized cost is the nominal cost plus the change in potential: $1+F_{i-4}-F_{i-1}+0 <0$ as $i \geq 6$.

\end{proof}

\subsection{Underflow-thru}\label{ss:ut}

 \figba{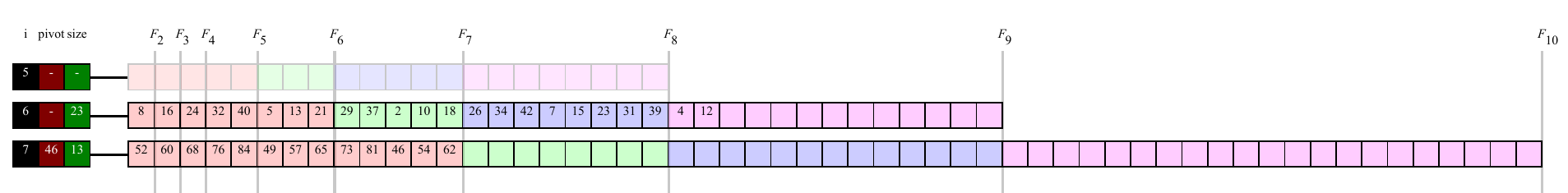}{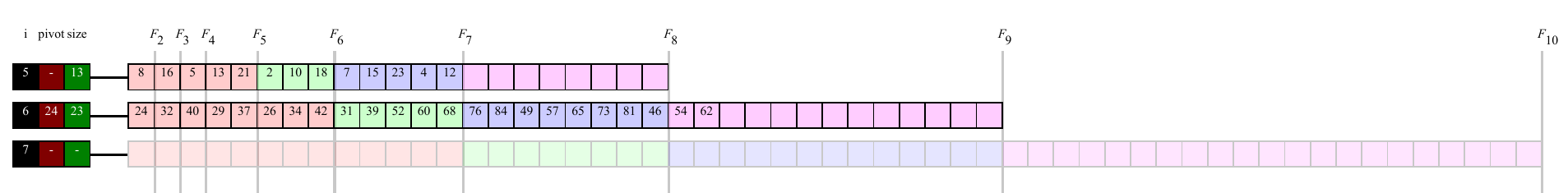}{An underflow-thru operation performed on $S_7$ which has $F_7=13$ elements when $S_5$ is nonempty.}{f:ut}

See Figure~\ref{f:ut}. The underflow-thru operation is called when a set $\jbefore{S}_i$ is underfull, and thus has $F_{i}$ items, and the set above, $\jbefore{S}_{i-1}$ is nonempty. 
As there can not be more than two non-empty consecutive sets by the consecutive nonempty invariant, this implies $\jbefore{S}_{i-2}$, is nonempty.

Recall that in this operation, $\jbefore{S}_i$ and $\jbefore{S}_{i-1}$ are combined, and the smallest items are moved from the combined set to $\jafter{S}_{i-2}$ so that the size of $S_{i-1}$ is unchanged, thus $|\jafter{S}_{i-1}|=|\jbefore{S}_{i-1}|$ and $|\jafter{S}_{i-2}|=|\jbefore{S}_{i}|=F_i$. This requires a linear time median/order statistic algorithm and thus has nominal cost $F_{i-6}$ since $S_{i}$ is the largest set involved. 
The pivots $\jafter{p}_{i-1}$ and $\jafter{p}_{i-2}$ are set to the minimum in each set.

\begin{lemma}
	The amortized cost of the underflow-thru operation is at most 0 when $i\geq 3$.
\end{lemma}

\begin{proof}
The nominal cost of the operation is $F_{i-6}$.	To compute the change in potential we look at each of the three components of the potential function:

\begin{description}
    \item[Change in $\potdown{}$:] 
    By Fact~\ref{f3} this is
    $\potdowna{i-1}+\potdowna{i-2}-\potdownb{i}-\potdownb{i-1}$. By Fact~\ref{f1}, $\potdownb{i}$ and $\potdowna{i-1}$ are both zero. Also, $-\potdownb{i-1}$ is at most 0. Thus the  change in $\potdownb{}$ is at most $\potdowna{i-2}$ which is $\maxzero{ F_{i-5}-\downsuma{i-2}} \leq F_{i-5}$.
    \item[Change in $\potsize{}$:] This is $\potsizea{i-2}+\potsizea{i-1}-\potsizeb{i-1}-\potsizeb{i}$.
    Note that while the contents of $S_{i-1}$ will change, the size will not and thus $\potsizeb{i-1}=\potsizea{i-1}$. Also note that by definition, the size of $\jafter{S}_{i-2}$ is $F_{i}$ and thus by definition $\potsizea{i-2}=0$. So, the change in $\potsize{}$ is $-\potsizeb{i}$:
\[
    	-\potsizeb{i}
    	-(F_{i+1}-\overbrace{|\jbefore{S}_i|}^{F_i})
    \leq -	(F_{i+1}-F_i)
    = -F_{i-1}    
\]
    \item[Change in $\potset{}$: $0$] The number of nonempty sets does not change. 
\end{description}
In summary, by definition, the amortized cost is the nominal cost plus the change in potential: 
$
F_{i-6}+F_{i-5}-F_{i-1}+0 <0
$.
\end{proof}

\subsection{\INSERTplain}\label{ss:ins}

In an \INSERT\ operation, a binary search is performed on the pivots to determine which set the newly inserted item is added to. 
If the item is smaller than the smallest pivot, it is added to the first set, $S_3$.
Denote this set as $j$.
The size of $j$ is updated and the pivot of $j$ is not updated as it is still valid.
 If this becomes overfull, a single overflow-down or overflow-thru is performed which eliminates the invariant violation on the set size but may cause too many empty sets in a row below $j$ and too many nonempty sets in a row above $j$. The former is fixed via split-up, and the latter via merge-down, which can cascade downwards. 

\begin{lemma}
	The amortized cost of the \INSERT\ operation is $O(\lg \lg n)$.
\end{lemma}

\begin{proof}
We need not analyze any invariant-restoring operations as these have been shown to have amortized cost at most zero (except for those called on constant sized sets which can happen at most a constant number of times with a constant amortized cost).

The nominal cost of this operation is $O(\lg \lg n)$ due to the binary search on the pivots of the $\Theta(\lg n)$ sets.

To compute the change in potential we look at each of the three components of the potential function. 

\begin{description}
\item[Change in $\potdown{i}$:] As this is defined as $ \maxzero{F_{i-3}-\downsum{i}}$, and the $\downsum{i}$ are only incremented if they change, this potential increases by at most 0.
    \item[Change in $\potsize{i}$:] $\leq 1$. The set $j$ now has one more element and thus could gain one unit of size potential.
    \item[Change in $\potset{i}$:] 0, as the number of sets do not change.
\end{description} 

In summary, by definition, the amortized cost excluding any invariant-restoring operations is the nominal cost plus the change in potential which is at most: 
$
O(\lg \lg n) + 0 + 1 + 0 = O(\lg \lg n)$.
Including any calls to invariant-restoring operations, which have a total of $O(1)$ amortized cost, gives the total 
amortized cost of $O(\lg \lg n)$.
\end{proof}

\subsection{\DELETEMINplain}\label{ss:em}

In the \DELETEMIN\ operation, we find and remove the smallest element from the smallest non-empty set, which we denote as $S_j$. Due to the consecutive empty invariant, $j\leq 11$ and thus this takes constant time as $|S_j|\leq F_{14}=377$.
We update the size of the set $j$.
If the set becomes underfull, we perform a single underflow-up or underflow-thru, which can cause a cascading chain of merge-downs to be performed.

\begin{lemma}
	The amortized cost of the \DELETEMIN\ operation is $O(\lg n)$.
\end{lemma}

\begin{proof}
We need not analyze any invariant-restoring suboperations as these have been shown to have amortized cost at most zero (except for those called on constant sized sets which can happen at most a constant number of times with a constant amortized cost).

The actual cost of this operation excluding the invariant-restoring sub-operations is $O(1)$.

To compute the change in potential we look at each of the three components of the potential function. 

\begin{description}

\item[Change in $\potdown{i}$:] As this is defined as $ \maxzero{F_{i-3}-\downsum{i}}$ for the nonempty sets (and 0 for the empty sets), and each of the $O(\lg n)$ $\downsum{i}$ will decrease by at most 1, the up potential increases by at most $O(\lg n)$.
    \item[Change in $\potsize{i}$:] $\leq 1$. The set $j$ now has one less element and thus could gain one unit of size potential.
    \item[Change in $\potset{i}$:] At most 0.
\end{description} 

In summary, by definition, the amortized cost is the actual cost plus the change in potential which is: 
$
\Oh{1} + \Oh{\lg n} + 1 + 0 = \Oh{\lg n} 
$.
Including any calls to invariant-restoring operations, which have a total of $\Oh{1}$ amortized cost, gives the total 
amortized cost of $\Oh{\lg n}$.
\end{proof}

\subsection{\DECREASEKEYplain}\label{ss:dc}

This operation is carried out by first removing the item from its current set, $S_j$, decreasing it, and re-inserting it in the appropriate set $S_k$; crucially $k\leq j$ as the key is only decreased. Two binary searches are needed on the previous value (to update the size of the set it was in) and the new set (to insert and also update the size). The pivots of both sets remain valid with no change needed. After the removal, if the set is underfull a single underflow-up or underflow-thru is performed, and then split-up and merge-downs are performed as needed until all invariants are restored; after the insertion if the set is overfull a single overflow-down or overflow-thru is performed, and then split-up and merge-down are performed as needed until all invariants are restored. 

\begin{lemma}
	The amortized cost of \DECREASEKEY\ is $\Oh{\lg \lg n}$.
\end{lemma}

\begin{proof}
We need not analyze any invariant-restoring suboperations as these have been shown to have amortized cost at most zero (except for those called on constant sized sets which can happen at most a constant number of times with a constant amortized cost).

The nominal cost of this operation is $O(\lg \lg n)$ due to the two binary search operations.

To compute the change in potential we look at each of the three components of the potential function. 

\begin{description}

\item[Change in $\potdown{i}$:] This is defined as $ \maxzero{ F_{i-3}-\downsum{i} }$ for the nonempty sets, and 0 for the empty sets. For $i<k$ and $i>j$ this will not change, but for $i$ between $k$ and $j$, that is for sets between the removal and re-insertion, $\downsum{i}$ will increase by 1, and thus $\potdown{i}$ will decrease by 1. This is where we crucially use the fact that keys can only decrease, as if we allowed an increase the potential here could grow by $\Theta(\lg n)$, but instead it may only shrink.
\item[Change in $\potsize{i}$:] $\leq 2$. The insertion and deletion could each cause a unit change in the size potential of the affected sets.
    \item[Change in $\potset{i}$:] 0, as the number of sets do not change.
\end{description} 

In summary, by definition, the amortized cost is the actual cost plus the change in potential which is at most: 
$
\Oh{\lg \lg n} + 0 + 2 + 0 = \Oh{\lg \lg n}$.
Including any calls to invariant-restoring operations, which have a total of $\Oh{1}$ amortized cost, gives the total 
amortized cost of $\Oh{\lg \lg n}$.
\end{proof}

}

\section{Heaps with Exponential Upper Bounded Sets}
\label{sec:exponential}

\newcommand{\phiPull}{\phi^{\mathrm{pull}}}
\newcommand{\phiPush}{\phi^{\mathrm{push}}}
\newcommand{\phiInsert}{\phi^{\mathrm{insert}}}
\newcommand{\phiPullAfter}{\after\phi^{\mathrm{pull}}}
\newcommand{\phiPushAfter}{\after\phi^{\mathrm{push}}}
\newcommand{\phiInsertAfter}{\after\phi^{\mathrm{insert}}}

In this section we consider another variant of a partition-based simple heap consisting of sets $S_1,\ldots,S_\ell$ and pivots $p_2\leq\cdots\leq p_\ell$.
Again, the time obtained for {\DELETEMIN} is amortized $\Oh{\lg n}$ and for {\INSERT} and {\DECREASEKEY} amortized $\Oh{\lg\lg n}$. 
In this variant we have exponential increasing upper bounds on the sizes of the sets, without any lower bounds, and allowing an arbitrary number of the sets to be empty. This is in contrast to the stronger structural invariants of the FH:TNG in \wref{sec:Fibonacci}.

\begin{description}
\item[{Invariant:}] $|S_i| < 3 \cdot 2^i$ for $1 \leq i \leq \ell$.
\end{description}

\noindent
The heap operations are implemented as follows:

\begin{itemize}
\item {\INSERT}($e$):\\ 
    Binary search over the pivots to find $i$ where
    $p_{i} \leq e < p_{i+1}$,
    and append $e$ to $S_i$. If
    $|S_i|=3\cdot 2^i$, perform a recursive \emph{push} of $S_i$ to
    $S_{i+1}$ (see below).

\item {\DELETEMIN}():\\
    If $S_1$ is empty, first recursively \emph{pull} elements
    into $S_1$ from $S_2$ (see below). Delete and return the minimum from $S_1$.
    If  $\ell>1+\lg n$, let $S_{\ell-1}$ be the concatenation of $S_{\ell-1}$ and $S_{\ell}$, 
    and discard $S_{\ell}$ and $p_{\ell}$ before returning, i.e., $\ell$ decreases by one (this ensures $\ell \leq 1+\lg n$).
    
\item {\DECREASEKEY}($\mathit{ptr}$, $\newkey$):\\ 
    Delete the element~$e$ pointed to by $\mathit{ptr}$ from its current 
    set $S_j$ (to decrement the size $|S_j|$ we find $j$ by a binary search over
    the pivots). Change the key of $e$ to $\newkey$,
    and perform a binary search over the pivots to find the set $S_i$ where to add $e$, and append $e$ to $S_i$.
    If $|S_i| = 3\cdot 2^i$, perform a recursive \emph{push} of $S_i$ to~$S_{i+1}$.
\end{itemize}

\paragraph{Push}

For {\INSERT} and {\DECREASEKEY}, the set $S_i$ receiving one more element might get size $|S_i| = 3 \cdot 2^i$. We \emph{push} all $3\cdot 2^i$
elements of $S_i$ to $S_{i+1}$, and set $S_i=\emptyset$. 
In general, a recursive push to $S_j$ of a set of elements $X$ 
always satisfies $2^{j-1} \leq |X| \leq 3\cdot 2^{j-1}$. 
If $|S_j|< 2^j$ before the push (i.e., $S_j$ is too small to be
pushed to $S_{j+1}$), $X$ is concatenated to $S_j$ in constant time, such that
$|S_j| < 5 \cdot 2^{j-1}$, and the push terminates. Otherwise, the set $X$ becomes $S_j$ and the old
$S_j$ is recursively pushed to~$S_{j+1}$. The pivots need to be updated
accordingly, and when $j=\ell+1$ a new set $S_{\ell+1}=X$ is
created, i.e., $\ell$ increases by one.

\paragraph{Pull}

The recursive \emph{pull} of elements into $S_1$, or more generally into $S_i$
when $S_i$ is empty, is performed as follows: If $S_{i+1}$ is empty, we first
recursively pull elements into $S_{i+1}$. If $|S_{i+1}|\leq 2^{i-1}$, we swap the
roles of $S_{i+1}$ and the empty~$S_i$, i.e., in constant time move all elements
from $S_{i+1}$ to $S_i$.
Otherwise, we select and delete the $2^{i-1}$ smallest elements from $S_{i+1}$ in time
$\Oh{2^i}$ and let $S_i$ be these (and update $p_{i+1}$; 
here it is crucial that elements are distinct to ensure $\max S_{i} < p_{i+1}=\min S_{i+1}$).

\subsection{Analysis}

To prove the amortized performance of this structure we apply the following 
potential function $\Phi = \phiInsert + \phiPush + \phiPull$, where
\[
    \phiInsert := \sum_{i=1}^{\ell} \phiInsert_{i}\quad\quad\quad
    \phiPush := \sum_{i=1}^{\ell} \phiPush_{i}\quad\quad\quad
    \phiPull := \sum_{i=1}^{\ell} \phiPull_{i}
\]
\[
    \phiInsert_{i} := \maxzero{|S_i| - 5 \cdot 2^{i-1}} \in [0, 2^{i-1}]
    \quad\quad\quad
    \phiPush_{i} := \left\lfloor |S_i| /2^{i} \right\rfloor \in [0, 3]
\]
\[
    \phiPull_{i} := \maxzero{2^{i-1} - \sum_{j=1}^i |S_j|} \in [0, 2^{i-1}]
\]

The basic idea is to let each set $S_i$ have a potential. A big set will have a push potential $\phiPush_i$ that releases one unit of potential when the set is moved to $S_{i+1}$. An insertion into $S_i$ can trigger a push of $S_i$. The saved up insertion potential $\phiInsert_i$ will be released by the push of $S_i$ to cover the increase in pull potential~$\phiPull_i$. Finally, pull potential is intuitively associated with each ``empty slot'' among the first $2^{i-1}$ ``slots'' in $S_i$, with a potential to cover the cost to be recursively filled from each of the subsequent sets $S_{i+1},\ldots,S_{\ell}$.
\LATIN{We refer the reader to the full version of the paper for a detailed amortized analysis.}
\ARXIV{%
Figure~\ref{fig:power-potentials} illustrates the relation between set sizes and the different potentials.

\begin{figure}[ht]
    \centering
    \begin{tikzpicture}[xscale=0.3, yscale=0.75]
        \foreach \c/\t [count=\i from 0] in {green/ideal size, orange/push potential, cyan/insertion potential} {
            \fill[\c] (28, -\i*0.75) rectangle ++(1, -0.5);
            \draw (29, -\i*0.75-0.25) node[right] {\small \t};
        }
        \foreach \i/\t in {0/S_1,1/S_2,2/S_3} {
            \node[left] at (0, -\i) {$\t$};
            \node at (2^\i, -\i+0.5) {\scriptsize $2^\i$};
            \node at (6*2^\i, -\i+0.5) {\scriptsize $6\!\cdot\!2^\i$};
            \ifthenelse{\equal{\i}{0}}{}{
                \node at (5*2^\i, -\i+0.5) {\scriptsize $5\!\cdot\!2^\i$};
                \node at (2*2^\i, -\i+0.5) {\scriptsize $2\!\cdot\!2^\i$};
            }
            \fill[green] (0, -\i+0.25) rectangle (2^\i, -\i-0.25);
            \fill[orange] (2^\i, -\i+0.25) rectangle (5*2^\i, -\i-0.25);
            \draw[dashed] (2*2^\i, -\i-0.25) -- ++(0, 0.5);
            \fill[cyan] (5*2^\i, -\i+0.25) rectangle (6*2^\i, -\i-0.25);
            \draw (0, -\i+0.25) rectangle (6*2^\i, -\i-0.25);
        }   
    \end{tikzpicture}
    \caption{Illustration of set sizes and potentials for the $S_i$ sets.}
    \label{fig:power-potentials}
\end{figure}
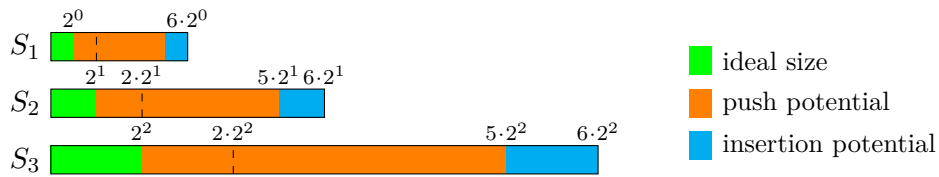

\begin{lemma}
\label{lem:power-ell}
    For a heap with $n$ elements $\ell \leq 1+\lg n$.
\end{lemma}
\begin{proof}
    $S_{\ell}$ is created when at least $2^{\ell-1}$ elements are pushed from $S_{\ell-1}$ to
    $S_{\ell}$ during {\INSERT}, i.e., $n\geq 2^{\ell-1}$ implying $\ell \leq 1+\lg n$.
    During {\DELETEMIN} $n$ decreases by one.
    The concatenation of $S_{\ell-1}$ and $S_{\ell}$ when $\ell > 1+\lg n$, ensures $\ell$ decreases by one and  $\ell \leq 1+\lg n$ after {\DELETEMIN}.
\end{proof}

\begin{lemma}
\label{lem:insert}
    An {\INSERT} operation, excluding recursive pushing, takes worst-case
    $\Oh{\lg \ell}$ time and the change in potential is $\Delta \Phi \leq 2$.
\end{lemma}
\begin{proof}
    The binary search for $i$ and appending to $S_i$ takes
    worst-case $\Oh{\lg \ell}$ time. Increasing $|S_i|$ by one can at most increase
    $\phiInsert_{i}$ and $\phiPush_{i}$ by one, and
    decrease $\phiPull_{j}$ for each $j \geq i$. It follows that
    $\Delta \Phi \leq 2$.
\end{proof}

\begin{lemma}
\label{lem:decrease-key}
    \sloppypar
    A {\DECREASEKEY} operation, excluding recursive pushing, takes
    worst-case $\Oh{\lg \ell}$ time and the change in potential is $\Delta \Phi \leq 2$.
\end{lemma}
\begin{proof}
    Performing {\DECREASEKEY} involves two binary searches on the $\ell-1$ pivots,
    requiring worst-case $\Oh{\lg \ell}$ time, and worst-case $\Oh{1}$ time for moving
    an element from $S_i$ to $S_j$, where $1 \leq j\leq i\leq \ell$.
    The potentials $\phiPull_{m}$ can only change for $j \leq m <
    i$, in which case they decrease by one,
    $\phiInsert_{i}$ and $\phiPush_{i}$ can only
    decrease, and $\phiInsert_{j}$ and
    $\phiPush_{j}$ at most increase by one. It follows that
    $\Delta\Phi \leq 2$.
\end{proof}

\begin{lemma}
\label{lem:push}
    The recursive pushing of $S_i$  into $S_{i+1}$, when $|S_i|=3 \cdot 2^i$, takes
    worst-case $\Oh{m-i}$ time and the change in potential is $\Delta \Phi \leq -m+i+3$, for some~$m$ where $i
    < m \leq \ell+1$.
\end{lemma}
\begin{proof}
    Assuming the recursive pushing continues until pushing elements into $S_m$
    for $i < m \leq \ell+1$, i.e., $|S_j| \geq 2^j$ for $i < j < m$ and
    $|S_m| < 2^m$ before the push. Each of the $m-i$ pushes takes
    worst-case constant time to either rename sets or to concatenate two
    doubly-linked lists.

    In the following, recall a bar over a value denotes the value after an operation.
    We first consider the change to the potential $\phiInsert$.
    Before the push, we have $|S_i|=6 \cdot 2^{i-1}$, i.e.,
    $\phiInsert_{i}=2^{i-1}$. Since after the push $|\after{S}_i|=0$, we have
    $\phiInsertAfter_{i}=0$. For $i+1 \leq j \leq m$, between $2^{j-1}$ and $3 \cdot 2^{j-1}$ elements
    are pushed to $S_j$. It follows that after the
    recursive push $|S_j| \leq 5 \cdot 2^{j-1}$, for $i\leq j\leq m$, i.e.,
    $\phiInsertAfter_{i}=\cdots=\phiInsertAfter_{m}=0$.
    It follows that $\Delta \phiInsert \leq -2^{i-1}$.
    
    For $\phiPull$, note that $|S_j|$ only changes for $i \leq j
    \leq m$ and only these $\phiPullAfter_{j}$ can change. For $i <
    j < m$, we have $|S_j| \geq 2^{j-1}$ before and after the push, i.e.,
    $\phiPull_{j}=\phiPullAfter_{j}=0$, $|S_m|$
    increases, i.e., $\phiPullAfter_{m}$ can only decrease, and
    the maximal increase in $\phiPull_{i}$ is from 0 to $2^{i-1}$, i.e.,
    $\phiPull \leq 2^{i-1}$.

    For $\phiPush$, we have $|\after S_i|=0$, i.e.,  
    $\phiPushAfter_{i}=0$. For $i < j < m$, since we have $\after S_j
    = S_{j-1}$ and $|S_{j-1}| \geq 2^{j-1}$ before the recursive push,
    we have $\phiPush_{j-1} \geq 1$ and
    $\phiPushAfter_{j} \leq \phiPush_{j-1}-1$. Since
    $\phiPushAfter_{i}=0$ and $\phiPushAfter_{m} \leq
    2$, it follows that $\Delta \phiPush \leq 2 - (m-i-1)=-m+i+3$.    
\end{proof}

\begin{lemma}
\label{lem:delete-min}
    A {\DELETEMIN} operation, excluding recursive pulling, takes
    worst-case $\Oh{1}$ time and the change in potential is $\Delta \Phi \leq \ell$.
\end{lemma}
\begin{proof}
    Removing the minimum from $S_1$ takes
    worst-case constant time, since $|S_1|\leq 6$. Decreasing $|S_1|$ can only
    decrease $\phiInsert_{1}$ and $\phiPush_{1}$, or
    leave them unchanged. All other $\phiInsert_{j}$ and
    $\phiPush_{j}$ remain unchanged. Decreasing $|S_1|$ by one can
    at most increase $\phiPull_{j}$ by one for all $1 \leq j \leq
    \ell$. The total change in potential becomes $\Delta\Phi \leq \ell$.

    Since the deletion of an element causes $n$ to decrease by one, we might have
    $\ell> 1+\lg n$. In this case {\DELETEMIN} in  worst-case constant time concatenates 
    $S_{\ell-1}$ and $S_{\ell}$, and $\ell$ decreases by one, ensuring $\ell \leq 1+\lg n$.
    Since $\ell> 1+\lg n$ before the concatenation,
    we have $|\after S_{\ell-1}|\leq n < 2^{\ell-1}$, and
    $\phiPush_{\ell-1}=\phiInsert_{\ell-1}=0$ and
    $\phiPull_{\ell-1}$ can only have decreased. It follows that $\Phi$
    can only decrease by performing this additional check during {\DELETEMIN}.
\end{proof}

\begin{lemma}
\label{lem:pull}
    The recursive pulling of elements into $S_1$, when $|S_1|=0$, takes
    worst-case $\Oh{m+2^{i-1}}$ time and the change in potential is $\Delta \Phi \leq -m-2^{i-1}+1$, for some $i$ and $m$, where $1 \leq i < m \leq \ell$.
\end{lemma}
\begin{proof}
    Assume when pulling into $S_1$ we have $|S_1|=\cdots=|S_{m-1}|=0$ and
    $|S_m|>0$ for some $1 < m \leq \ell$. Furthermore, let $i$ be the largest $i <
    m$ where $2^{i-1} < |S_m|$. I.e., $S_m$ will be iteratively swapped with
    $S_{m-1},\ldots,S_{i+1}$, and for $1 \leq j \leq i$ we fill $S_j$ by
    selecting the $2^{j-1}$ smallest elements from $S_{j+1}$ in $\Oh{2^j}$ time. It
    follows that the worst-case time is $\Oh{m+\sum_{j=1}^i 2^j} =
    \Oh{m+2^i}$.
    
    After the pull, $|S_j|\leq 2^{j-1}$ for $1
    \leq j < m$, i.e.,
    $\phiInsertAfter_{j}=0$ and $\phiPushAfter_{j}=0$.
    Since $|S_m|$ decreases, we have that $\Delta \phiInsert_{m} \leq 0$
    and $\Delta \phiPush_{m} \leq 0$. 
    In total $\phiInsert$ and $\phiPush$ can only decrease. 
    
    For $\phiPull$, we have that
    $\phiPullAfter_{m},\ldots,\phiPullAfter_{\ell}$ are
    unchanged, since we only move elements from $S_m$ to $S_0,\ldots,S_{m-1}$.
    Since at least one element is moved to $S_0$, we have
    $\phiPull_{j}$ decreases by at least one for all $1 \leq j <
    m$. For $1 \leq j \leq i$, exactly $2^j$ elements are moved to $S_j$, and subsequently
    subsets hereof to $S_{j-1}, \ldots, S_1$, i.e., we have
    $\Delta\phiPull_{j}=-2^{j-1}$.  It follows that
    $\Delta\phiPull \leq -\sum_{j=1}^{i} 2^{j-1} - (m - 1 - i)  =
    -2^{i}-m+i+2\leq -2^{i-1}-m+1$, which is also a bound on $\Delta \Phi$.
\end{proof}

By \wref{lem:power-ell} we have $\ell=\Oh{\lg n}$. Lemmas~\ref{lem:power-ell}, \ref{lem:insert} and
\ref{lem:push} imply {\INSERT} takes amortized $\Oh{\lg\lg n}$ time.
Lemmas~\ref{lem:power-ell}, \ref{lem:decrease-key} and \ref{lem:push} imply {\DECREASEKEY}
takes amortized $\Oh{\lg\lg n}$ time. Finally, Lemmas~\ref{lem:power-ell},
\ref{lem:delete-min} and \ref{lem:pull} imply {\DELETEMIN} takes amortized
$\Oh{\ell}=\Oh{\lg n}$ time.

\begin{corollary}
    The amortized time for {\INSERT} and {\DECREASEKEY} is $\Oh{\lg\lg n}$ and the amortized time for {\DELETEMIN} is $\Oh{\lg n}$.
\end{corollary}

}%

% --------------------------------------------------------------------
\section{Conclusion}
\label{sec:open_problems}
% --------------------------------------------------------------------

%

We presented three priority queue implementation, following a simple partition-based heap framework
that uses elementary data structures and algorithms.
In particular for the simplest design, the \CSHeapsLong, we give a fully self-contained and complete description and analysis.

Replacing the binary search on an array by a binary search tree, the implementation works in the pointer-machine model and requires only constant indegree on nodes, and can hence be made persistent efficiently.

Due to its simplicity and linear structure, our implementation appears amenable to adaptations for external memory, but supporting pointer-based decrease-key operations becomes more complicated (as noted by~\cite{NavarroParedesPobleteSanders2011}).
We reserve this for future work.

Using our rule as the concatenation rule in quickheaps promises to yield a highly practical and robust priority queue implementation, combining the best of previously known quickheaps variants.
A detailed empirical evaluation is planned for future work.

The astute reader may have noticed that {\INSERT} and {\DECREASEKEY} take amortized constant time, except for the binary searches over the $\Oh{\lg n}$ pivots. By using a non-pointer-machine model data structure such as a fusion tree~\cite{DBLP:journals/jcss/FredmanW93} to store the pivots, the running time can be reduced to $\Oh{\lg_w \lg n}$, which under the standard assumption that $w=\Omega(\log n)$ results in $\Oh{1}$ amortized running times for these operations. 

%

%\ifkoma{
\myacknowledgements
%}{}

%\clearpage
% !BIB program = bibtex
\bibliography{references,bib}

\ifproceedings{}{
\ARXIV{
    \clearpage
    \appendix
    \ifkoma{\addpart{Appendix}}{}
    \section{LP Heap Proof-of-Concept Implementation}

To illustrate the simplicity of LP Heaps, we provide a fully self-contained Python implementation below.
The code is intended to serve as detailed specification of LP Heaps, tuned for readability.
It is not intended to be efficient.

We first show the main code for LP Heaps, the class \texttt{LazyHeap}; it uses a doubly-linked list and a selection method which we provide below.

\begin{lstlisting}[language=Python,gobble=4]
	class LazyHeap:
	    def __init__(self):
	        self.counter = 0
	        self.pivots = []
%	        self.sets = [LinkedList()]
	
	    def insert(self, value):
	        node = Node((value, self.counter))
	        self.counter += 1
	        self._insert(node)
	        return node
	
	    def delete(self, node):
	        self._remove(node)
	        self._merge_sets()
	
	    def extract_min(self):
	        S = self.sets[0]
	        if len(S) == 0:
	            raise ValueError("Heap is empty")
	        smallest_node = min(S.iter_nodes(), 
	                            key=lambda node: node.value)
	        S.remove(smallest_node)
	        if len(S) >= 2:
	            median, first, second = self._partition_set(S)
	            self.pivots = [median] + self.pivots
	            self.sets = [first, second] + self.sets[1:]
	        self._forget_pivots()
	        return smallest_node.value[0]
	
	    def decrease_key(self, node, new_value):
	        assert new_value <= node.value[0]
	        self._remove(node)
	        node.value = (new_value, node.value[0])
	        self._insert(node)
	
	
	    def _insert(self, node):
	        S = self._find_set(node.value)
	        S.append(node)
	
	    def _remove(self, node):
	        S = self._find_set(node.value)
	        S.remove(node)
	
	
	    def _find_set(self, value):
	        if len(self.sets) == 1 or value < self.pivots[0]:
	            return self.sets[0]
		        low = 0
	        high = len(self.pivots)
	        while low + 1 < high:
	            mid = (low + high) // 2
	            if self.pivots[mid] <= value:
	                low = mid
	            else:
	                high = mid
	        assert value >= self.pivots[low]
	        return self.sets[low + 1]
	
	    def _partition_set(self, S):
	        median = select(S, -(-len(S) // 2))
	        smaller = []
	        larger = []
	        for node in S.iter_nodes():
	            if node.value < median:
	                smaller.append(node)
	            else:
	                larger.append(node)
	        first = LinkedList()
	        for node in smaller:
	            first.append(node)
	        second = LinkedList()
	        for node in larger:
	            second.append(node)
	        return median, first, second
	
	    def _forget_pivots(self):
	        new_pivots = []
	        new_sets = [self.sets[0]]
	        outside = 0
	        for pivot, S in zip(self.pivots, self.sets[1:]):
	            if len(new_sets[-1]) == 0 or 
	                    outside > len(new_sets[-1]) + len(S):
	                new_sets[-1].concat(S) // merge
	            else:
	                outside += len(new_sets[-1])
	                new_pivots.append(pivot)
	                new_sets.append(S)
	        if len(new_sets[-1]) == 0 and len(new_sets) >= 2:
	            new_pivots.pop()
	            new_sets.pop()
	        self.pivots = new_pivots
	        self.sets = new_sets
	
\end{lstlisting}

The code above uses a circularly closed doubly-linked list;
since the node objects are used as pointers/references to inserted elements,
we provide the simple implementation here:

\begin{lstlisting}[language=Python,gobble=4,deletekeywords={next}]
	class Node:
	    def __init__(self, value):
	        self.value = value
	        self.prev = None
	        self.next = None
	
	class LinkedList:
	    def __init__(self):
	        self.head = Node(None)
	        self.tail = Node(None)
	        self.head.next = self.tail
	        self.tail.prev = self.head
	        self.size = 0
	    
	    def __len__(self):
	        return self.size
	
	    def append(self, node):
	        self.tail.prev.next = node
	        node.prev = self.tail.prev
	        node.next = self.tail
	        self.tail.prev = node
	        self.size += 1
	    
	    def remove(self, node):
	        node.prev.next = node.next
	        node.next.prev = node.prev
	        node.prev = None
	        node.next = None
	        self.size -= 1
	    
	    def concat(self, other):
	        t = self.tail.prev
	        h = other.head.next
	        t.next = h
	        h.prev = t
	        self.tail = other.tail
	        self.size += other.size
	
	    def iter_nodes(self):
	        at = self.head.next
	        while at is not self.tail:
	            yield at
	            at = at.next
	
	    def __iter__(self):
	        for node in self.iter_nodes():
	            yield node.value
\end{lstlisting}

To be self-contained, we also give a simple implementation of the median-of-medians linear-time deterministic selection method~\cite{BlumFloydPrattRivestTarjan1973}.
Note that the method is not in-place and should be considered a ``pedagogical'' implementation.
(A practical implementation would certainly want to use a randomized quickselect variant instead.)

\begin{lstlisting}[language=Python,gobble=4]
	def select(iterable, k):
	    def inner(arr, k):
	        buckets = [sorted(arr[i:i+5]) for i in range(0, len(arr), 5)]
	        
	        if len(buckets) == 1:
	            return buckets[0][k]
	
	        medians = [buck[len(buck) // 2] for buck in buckets]
	        pivot = inner(medians, len(medians) // 2)
	
	        smaller = []
	        larger = []
	
	        for value in arr:
	            if value < pivot:
	                smaller.append(value)
	            else:
	                larger.append(value)
	        
	        if k < len(smaller):
	            return inner(smaller, k)
	        else:
	            return inner(larger, k - len(smaller))
	
	    return inner(list(iterable), k)
\end{lstlisting}

}
}

\end{document}